\documentclass[twoside,11pt]{article}

\usepackage{blindtext}

\usepackage{jmlr2e}

\usepackage{amssymb}

\usepackage{paralist, amsmath, amssymb, color, graphicx, hyperref}
\usepackage[margin=1in]{geometry}
\DeclareGraphicsExtensions{.jpg,.pdf,.eps}
\usepackage{amssymb}
\usepackage{setspace}
\usepackage{mathtools}
\usepackage{color}
\usepackage{cancel}
\usepackage{todonotes}
\usepackage{algorithm,algorithmic}
\usepackage{verbatim}
\usepackage{thm-restate}
\newtheorem{example}{Example} 
\newtheorem{theorem}{Theorem}
\newtheorem{lemma}[theorem]{Lemma} 
\newtheorem{proposition}[theorem]{Proposition} 
\newtheorem{remark}[theorem]{Remark}

\newtheorem{definition}[theorem]{Definition}

\usepackage{stmaryrd}

\usepackage{xfrac}

\usepackage{mathrsfs}
\usepackage{amscd}
\usepackage{dsfont}
\usepackage{tikz}

\usepackage{enumerate}

\usepackage{listings}
\usepackage{color} 
\definecolor{mygreen}{RGB}{28,172,0} 
\definecolor{mylilas}{RGB}{170,55,241}
\definecolor{mygray}{gray}{0.95}

\newcommand{\mscr}[1]{\mathscr{#1}}

\newcommand{\ZZ}{\mathbb{Z}}
\newcommand{\RR}{\mathbb{R}}

\newcommand{\EE}{\mathbb{E}}

\newcommand{\de}{\delta}
\newcommand{\ep}{\epsilon}

\newcommand{\De}{\Delta}

\usepackage{mathtools}

\newcommand{\defeq}{\vcentcolon=}

\DeclareMathOperator{\var}{Var}

\usepackage{lastpage}
\jmlrheading{xx}{xxxx}{1-\pageref{LastPage}}{x/x; Revised x/x}{x/x}{x-xxxx}{Jordan Awan and Aishwarya Ramasethu}

\ShortHeadings{Optimizing Noise for \texorpdfstring{$f$}{f}-DP}{Awan and Ramasethu}
\firstpageno{1}

\begin{document}

\title{Optimizing Noise for \texorpdfstring{$f$}{f}-Differential Privacy\\
via Anti-Concentration and Stochastic Dominance}

\author{\name Jordan Awan \email jawan@purdue.edu \\
       \addr Department of Statistics\\
       Purdue University\\
       West Lafayette, IN 47907, USA
       \AND
       \name Aishwarya Ramasethu \email aramaset@purdue.edu \\
       \addr Department of Statistics\\
       Purdue University\\
       West Lafayette, IN 47907, USA}

\editor{}

\maketitle
\thispagestyle{empty}
\begin{abstract}
In this paper, we establish anti-concentration inequalities for additive noise mechanisms which achieve $f$-differential privacy ($f$-DP), a notion of privacy phrased in terms of a tradeoff function $f$ which limits the ability of an adversary to determine which individuals were in the database. We show that \emph{canonical noise distributions} (CNDs), proposed by \citet{awan2023canonical}, match the anti-concentration bounds at half-integer values, indicating that their tail behavior is near-optimal. We also show that all CNDs are sub-exponential, regardless of the $f$-DP guarantee. In the case of log-concave CNDs, we show that they are the stochastically smallest noise compared to any other noise distributions with the same privacy guarantee. In terms of integer-valued noise, we propose a new notion of discrete CND and prove that a discrete CND always exists, can be constructed by rounding a continuous CND, and that the discrete CND is unique when designed for a statistic with sensitivity 1. We further show that the discrete CND at sensitivity 1 is stochastically smallest compared to other integer-valued noises. Our theoretical results shed light on the different types of privacy guarantees possible in the $f$-DP framework and can be incorporated in more complex mechanisms to optimize performance.

\end{abstract}


\section{Introduction}
Differential privacy (DP), introduced by \citet{dwork2006calibrating}, is the state-of-the-art framework for formal privacy protection. DP methods require the introduction of noise into data analyses, which obscures the contribution of any particular individual.  Since its inception, DP has grown in popularity and is now employed by leading tech giants like Google \citep{erlingsson2014rappor} and Apple \citep{tang2017privacy}, as well as by the US Census Bureau \citep{abowd2018us}.

While there have been several variants of DP proposed, they all quantify the privacy risk in terms of a similarity measure between the distributions of outputs, when the mechanism is run on two databases differing in a single individual's contribution. The differences between these variants are primarily in how the ``similarity'' is measured. Recently, \citet{dong2022gaussian} proposed $f$-DP, which is rooted in hypothesis testing. The $f$ parameter in $f$-DP is a function which offers a more expressive quantification of the privacy risk compared to other notions of DP. The $f$-DP framework has the benefit of lossless conversions to the popular $(\ep,\de)$-DP, as well as divergence-based notions of DP \citep{bun2016concentrated,mironov2017renyi}.

The most basic technique for achieving DP is independent noise addition to a statistic of interest. Popular examples of such noise distributions are Laplace and Gaussian. In this paper, we investigate the limitation of what noise can be used to satisfy $f$-DP, and optimize the noise distributions in various settings.\\ 

\noindent \textbf{Our Contributions:} 
To understand the limits on the magnitude of noise required to satisfy $f$-DP, we develop an anti-concentration inequality, which gives an upper limit on the concentration of the noise distribution about its center. This upper bound is determined by the total variation distance between the noise distribution and a shifted version of itself, which can be bounded in terms of the $f$-DP guarantee.

We show that canonical noise distributions (CNDs), proposed by \citet{awan2023canonical} match our anti-concentration bound, which implies that their tail behavior is near-optimal. We also show that all CNDs are sub-exponential, indicating that it is never necessary to use heavier-tailed distributions to achieve $f$-DP. In the special case of log-concave CNDs, a subclass previously investigated by \citet{awan2022log}, we show that the log-concave CND is stochastically smaller than any other additive noise with the same privacy guarantee.

Lastly, we propose a new concept of discrete CND, which generalizes CNDs to integer-valued noises. Unlike continuous CNDs, which are solely defined in terms of a tradeoff function $f$, a discrete CND is also defined for a specific sensitivity. We show that a discrete CND always exists for any $f$ and any sensitivity, and that one can be constructed by rounding a continuous CND to integer values. In the case that the sensitivity is 1, we prove that there is a unique discrete CND for each tradeoff function and that this noise mechanism is stochastically smaller than any other integer-centered discrete noise satisfying $f$-DP.\\

\noindent \textbf {Related Work: }\label{s:related}
Various ways to optimize noise for differential privacy have been previously explored. 
\citet{ghosh2012universally} show that for a wide variety of loss functions, the geometric mechanism (discrete Laplace) minimizes the expected Bayesian loss under $\epsilon$-DP, regardless of what prior is used. \citet{gupte2010universally}  prove a similar result, showing that the geometric mechanism optimizes the minimax loss under $\epsilon$-DP.  \citet{qin2022differential2} consider necessary conditions for general discrete distributions to satisfy either $(\ep,0)$-DP or $(\ep,\de)$-DP. In the extended version of their paper,  they optimize the Wasserstein distance between the mechanism's input and output distributions, and show that a discrete staircase distribution is optimal \citep{qin2022differential}. 

\citet{geng2015optimal} focus on the minimization of loss functions in $(\ep,0)$-DP, with particular emphasis on the mean absolute error and mean squared error losses. They established an optimal noise-adding mechanism which they call the staircase mechanism. \citet{geng2015approx} consider optimizing the same loss functions in $(\ep,\de)$-DP, and found that in the high privacy regime, the discrete Laplace and uniform noise were nearly optimal for the mean absolute error and mean squared error losses.
\citet{SORIACOMAS2013200} construct an optimal data-independent noise-adding mechanism under $\ep$-DP, where the optimality is measured by moving the most probability mass towards zero. Their solution results in a staircase distribution, similar to \citet{geng2015optimal}.
In the setting of $\ep$-local-DP, a stronger form of differential privacy than the one considered in this paper, \citet{kairouz2016extremal} showed that for any sub-linear utility function, there exists a staircase mechanism which maximizes this utility.

The literature above is focused on finding the optimal noise under either $(\ep,0)$-DP or $(\epsilon,\delta)$-DP. $f$-Differential Privacy \citep{dong2022gaussian} is a more generalized privacy framework. \citet{awan2023canonical} was the first work to build a noise-adding mechanism for arbitrary $f$-DP, which they called a canonical noise distribution (CND). 
CNDs were designed to tightly match the $f$-DP privacy bound and were shown to lead to the most powerful hypothesis tests on binary data \citep{awan2023canonical};  They have also seen applications in other DP testing problems \citep{awan2022differentially,kazan2023test}. \citet{awan2022log} extended the results on CNDs by establishing the existence and construction of both log-concave CNDs and multivariate CNDs.\\

\noindent \textbf{Organization:} 
 In Section \ref{s:background}, we review the basics of $f$-differential privacy  and define the notion of optimality considered in this paper. In Section \ref{s:anti}, we develop our general anti-concentration inequality for additive noise. In Section \ref{s:continuous}, we consider the case of continuous noise. We review background on CNDs in Section \ref{s:CNDback}, show that CNDs are near-optimal in Section \ref{s:nearCND}, establish that CNDs are sub-exponential in Section \ref{s:subExp}, and derive optimality results for log-concave CNDs in Section \ref{s:logCND}. In Section \ref{s:discrete}, we introduce the concept of discrete CNDs. In Section \ref{s:discreteBasics}, we establish existence and construction of discrete CNDs, as well as the uniqueness of the discrete CND at sensitivity 1. In Section \ref{s:discreteOpt}, we prove the optimality of the discrete CND at sensitivity 1. We conclude with discussion in Section \ref{s:discussion}. Technical details and proofs are postponed to the appendix.\\

\section{Differential Privacy Background}\label{s:background}

Differential privacy is a framework that ensures that an adversary cannot accurately determine whether an individual's data is present in a database, based on the output of a privacy mechanism. A privacy mechanism $M$ is a randomized algorithm that takes as input a database $D$ and outputs a random variable $M(D)$. We write $M:\mscr D\rightarrow \mscr Y$ to indicate that it takes in a database $D\in \mscr D$ and the random variable $M(D)$ takes values in $\mscr Y$, where $\mscr D$ is the space of possible databases. A mechanism $M$ satisfies differential privacy if for two databases $D$ and $D'$,  differing in one entry, formalized in terms of an \emph{adjacency metric} $d(D,D')\leq 1$, the distributions of $M(D)$ and $M(D')$ are similar. While the similarity of $M(D)$ and $M(D')$ was previously formalized in terms of divergences on probability measures, \citet{dong2022gaussian} established that it is most natural to use the language of hypothesis testing, using the concept of a \emph{tradeoff function}.

For two distributions $P$ and $Q$, the \emph{tradeoff function} between $P$ and $Q$ is $T(P,Q):[0,1]\rightarrow [0,1]$, where $T(P,Q)(\alpha)=\inf \{1-\EE_Q (\phi) \mid \EE_P(\phi)\leq 1-\alpha\}$, where the infimum is over all measurable tests $\phi$. The tradeoff function returns the smallest type II error for testing $H_0=P$ versus $H_1=Q$ at specificity $\geq \alpha$ (i.e., $1-\alpha$ is an upper bound on the type I error), and captures the difficulty of distinguishing between $P$ and $Q$. The tradeoff function is closely related to the \emph{receiver operator characteristic curve} (ROC curve), which  evaluates sensitivity/power (one minus type II error) as a function of type I error.\footnote{In \citet{dong2022gaussian}, the tradeoff function was originally defined as the smallest type II error as a function of type I error. Our choice to flip the tradeoff function along the $x$-axis follows that of \citet{awan2022log} and is for mathematical convenience. There is a typo in \citet{awan2022log} where they mistakenly have the inequality in the wrong direction in the definition of $T(P,Q)$.}

We say that a function $f:[0,1]\rightarrow [0,1]$ is a tradeoff function if there exist distributions
$P$ and $Q$ such that $f(\alpha)=T(P,Q)(\alpha)$ for all $\alpha\in [0,1]$. A function $f:[0,1]\rightarrow [0,1]$ is a tradeoff function if and only if $f$ is convex, continuous, non-decreasing, and $f(x) \leq x$ for all $x \in [0,1]$ \citep[Proposition 2.2]{dong2022gaussian}.  We say that a tradeoff function $f$ is \emph{nontrivial} if $f(\alpha)<\alpha$ for some $\alpha\in (0,1)$.

\begin{definition}[$f$-DP: \citealp{dong2022gaussian}]\label{def:fDP}
 Let $f$ be a tradeoff function. A mechanism $M$ satisfies $f$-DP if $T(M(D),M(D'))\geq f,$  for all $D,D'\in \mscr D$ such that $d(D,D')\leq 1$.
\end{definition}

If $M$ satisfies $f$-DP, where $f=T(P,Q)$, then intuitively determining whether the database is $D$ or $D'$ is at least as hard as distinguishing between $P$ and $Q$. Thus, if $P$ and $Q$ are difficult to distinguish, then the individual's private data is unlikely to be compromised. Since the differential privacy guarantee is symmetric (i.e., both $T(M(D),M(D'))$ and $T(M(D'),M(D))$ must be greater than $f$), it suffices to restrict the $f$-DP guarantee to \emph{symmetric} tradeoff functions, meaning that if $f=T(P,Q)$, then $f=T(Q,P)$. Due to this, we limit the scope of the paper to symmetric tradeoff functions.

While $f$-DP consists of a wide variety of privacy guarantees, depending on the tradeoff function chosen, two important sub-families are worth pointing out. The common $(\ep,\de)$-DP version of differential privacy is a special case of $f$-DP: let $\ep\in [0,\infty)$ and $\de\in [0,1]$, then $M$ satisfies $(\ep,\de)$-DP if and only if it satisfies $f_{\ep,\de}$-DP, where \[f_{\ep,\de}(\alpha)= \max\{0,1-\de-\exp(\ep)+\exp(\ep)\alpha,\exp(-\ep)(\alpha-\de)\}.\] When $\de=0$, we call $(\ep,0)$-DP \emph{pure DP}, which was the original definition of differential privacy. When $\ep=0$, then $f_{0,\de}$-DP can be thought of as $\de$-Total Variation DP, as it is equivalent to bounding the total variation distance between $M(D)$ and $M(D')$ by $\delta$; recall that for two distributions $P$ and $Q$, the \emph{total variation distance} is $\mathrm{TV}(P,Q) = \sup_{A} |P(A)-Q(A)|$.

Another sub-family of $f$-DP is Gaussian-DP (GDP). We say that $M$ satisfies $\mu$-GDP if it satisfies $G_\mu$-DP, where $G_\mu = T(N(0,1),N(\mu,1))$. Besides offering the intuitive appeal of being defined in terms of shifted Gaussians, GDP also has a number of useful technical properties: the family $G_\mu$-DP is closed under group privacy and composition, and these operations commute under $G_\mu$-DP. Furthermore, \citet{dong2022gaussian} established a ``central limit theorem for composition,'' which shows that under mild conditions the composition of many DP mechanisms asymptotically satisfies $\mu$-GDP for some $\mu$.  

For a symmetric tradeoff function $f$, denote by $c_f$ the unique value such that $f(1-c_f)=c_f$, which can be visualized as the intersection of $f$ with the line $1-\alpha$. The value $c_f$ is a central concept to many of the results in this paper; we highlight a few basic properties here.  

\begin{restatable}{lemma}{lemCfzero}\label{lem:Cfzero}
Let $f$ be a symmetric tradeoff function. Then,
\begin{enumerate}
    \item $c_f\in [0,1/2]$ and if $f$ is nontrivial then $c_f\in [0,1/2)$,
    \item if $f=T(P,Q)$, then $\mathrm{TV}(P,Q)=1-2c_f$,
    \item $f_{0,(1-2c_f)}\leq f\leq f_{\ep_f,0}$, where $\ep_f=\log\left(\frac{1-c_f}{c_f}\right)$.
\end{enumerate}
\end{restatable}
Property 2 of Lemma \ref{lem:Cfzero} makes an explicit connection between the tradeoff function and the total variation distance, which will be heavily used throughout the paper. By property 3 of Lemma \ref{lem:Cfzero}, for a given symmetric tradeoff function $f$, we can construct nontrivial upper and lower bounds for $f$ using only the value $c_f$. In particular, it shows that for any $f$, there is a pure-DP guarantee that implies $f$-DP, namely $(\epsilon_f,0)$-DP where $\ep_f = \log((1-c_f)/c_f)$.

Similar to other notions of differential privacy, $f$-DP has a number of important properties:\\

\noindent\textbf{Group Privacy:} The guarantee of $f$-DP is implicitly for \emph{groups of size 1}, since the adjacent databases differ in one entry. However, if $M$ satisfies $f$-DP, then it also satisfies $f^{\circ k}$-DP for groups of size $k$ (where $d(D,D')\leq k$), where $f^{\circ k}\defeq f\circ \cdots \circ f$ and $f$ appears $k$ times. Note that while this guarantee is tight in general, for a specific $f$-DP mechanism there may exist a stronger guarantee for groups. \\

\noindent\textbf{Composition:} Another property of $f$-DP is that it can quantify the cumulative privacy cost of multiple privacy mechanisms. Specifically, if $M_1$,\ldots, $M_k$ each satisfy $f_1$-DP,\ldots, $f_k$-DP, respectively,  and $(M_j(D))_j$ are mutually independent given $D$, then the joint release of $(M_1(D),\ldots, M_k(D))$ satisfies $f_1\otimes \cdots \otimes f_k$-DP, where ``$\otimes$'' is the \emph{tensor product} of tradeoff functions: if $f=T(P_1,Q_1)$ and $g=T(P_2,Q_2)$, then $f\otimes g=T(P_1\times P_2,Q_1\times Q_2)$. See \citet{dong2022gaussian} for a more complex sequential composition result, where each $M_i(D)$ may depend on the results of $M_1(D),\ldots, M_{i-1}(D)$.\\

\noindent\textbf{Invariance to Postprocessing:} If $M: \mscr D\rightarrow \mscr Y$ satisfies $f$-DP and $T:\mscr Y\rightarrow \mscr Z$ is a (possibly randomized) function, then $T\circ M:\mscr D\rightarrow Z$ also satisfies $f$-DP. This property is important to ensure that the privacy guarantee cannot be ``undone'' by a data-independent attack. \\

\noindent\textbf{Additive Noise Mechanisms:} The class of privacy mechanisms that we study in this paper are \emph{ additive noise mechanisms}, which add noise scaled to the \emph{sensitivity} of a statistic. A statistic $S:\mscr D\rightarrow \RR$ has sensitivity $\Delta>0$ if $|S(D)-S(D')|\leq \Delta$ for all $d(D,D')\leq 1$. An additive $f$-DP mechanism at sensitivity $\Delta$ is a random variable $N\sim F$ which satisfies $T(S(D)+N,S(D')+ N)\geq f$ for all $d(D,D')\leq 1$  and all statistics $S$ with sensitivity $\Delta$. More succinctly, this is equivalent to $T(N,N+m)\geq f$ for all $m\in [-\Delta,\Delta]\cap \mathrm{Range}(S)$. If $N$ is a continuous noise, then we can rescale both $S$ and $N$ by $\Delta$. 

Two of the most common  additive noise mechanisms are the Laplace and Gaussian mechanisms: adding $\mathrm{Laplace}(0,\Delta/\ep)$ to a real-valued statistic $S$ with sensitivity $\Delta$ satisfies $(\ep,0)$-DP, and adding $N(0, \Delta^2/\mu^2)$ to the same statistic satisfies $\mu$-GDP. 

 While elementary, additive noise mechanisms are widely used in DP, either by themselves or as part of more complex mechanisms. For example, a count statistic has sensitivity 1, the sample mean of $n$ data points lying in $[a,b]$ has sensitivity $(b-a)/n$, and the sample variance of $n$ data points lying in $[a,b]$ has sensitivity $(b-a)^2/n$ \citep{du2020differentially}; all of these statistics can be privatized to achieve $f$-DP by using an additive noise mechanism scaled to the sensitivity.  \\

\noindent\textbf{Stochastic Dominance and Optimal Noise Mechanisms:} Let $X$ and $Y$ be two random variables. Stochastic dominance of random variables is a partial order, which asserts that one variable takes larger values compared to another. The following are equivalent definitions of (first order) stochastic dominance \citep{levy1992stochastic}: 1) $X$ stochastically dominates $Y$, 2) $F_X(t)\leq F_Y(t)$ for all $t\in \RR$, 3) for every non-decreasing function $\phi$, $\EE \phi(X)\geq \EE\phi(Y)$, 4) there exists a random variable $W$ (not necessarily independent of $Y$) such that $X\overset d = Y+W$. 

If $N$ and $N'$ are two additive noise mechanisms, we say that $N$ is \emph{stochastically smaller} than $N'$ if $|N'|$ stochastically dominates $|N|$: $P(|N|\leq t)\geq P(|N'|\leq t)$ for all $t\in \RR^{\geq 0}$. If we have that $P(|N|\leq t)\geq P(|N'|\leq t)$ for all $N'$ in some class of noise distributions $\mscr A$, we say that $N$ is the \emph{stochastically smallest} noise distribution in $\mscr A$ or simply that $N$ is \emph{optimal} in $\mscr A$.


\section{A General Anti-Concentration Inequality}\label{s:anti}
Concentration inequalities such as Markov, Chebyshev, and Bernstein, give upper bounds on the amount of mass in the tails of a distribution. In this section, we establish \emph{anti-concentration inequalities} for additive $f$-DP mechanisms, which instead give upper bounds on how much mass is near the center of the distribution. Our inequalities establish these bounds in terms of total variation distance and $f$-DP tradeoff functions. 
 The most basic version of our anti-concentration inequality is in Lemma \ref{lem:anti}, which bounds the amount of mass in the center of the distribution in terms of the total variation distance between the random variable and a shifted version of itself. 

\begin{figure}
    \centering
    \includegraphics[width=.8\linewidth]{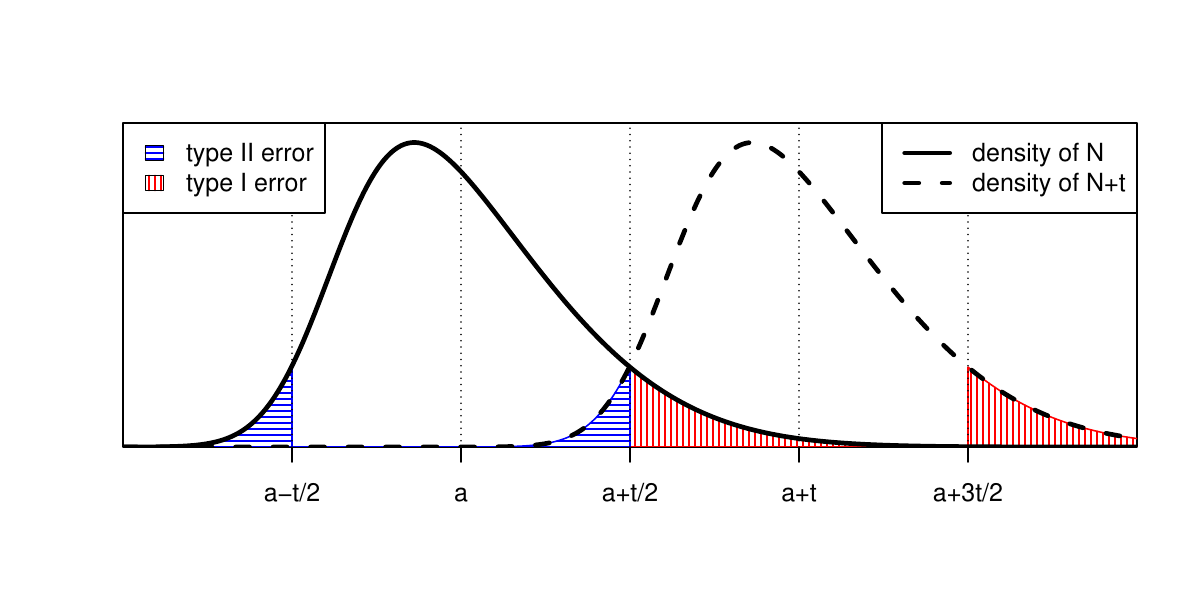}
    \caption{Illustration of Lemma \ref{lem:anti}, and the connection to hypothesis testing. When testing $H_0: N$ versus $H_1: N+t$, and using the rejection region $(a+t/2,\infty)$, the type I error probability is illustrated by the vertically shaded red region and the type II error probability is given by the horizontally shaded blue region. }
    \label{fig:antiLemma}
\end{figure}

\begin{restatable}{lemma}{lemanti}
    [Anti-Concentration Inequality]\label{lem:anti}
    Let $N$ be a real-valued random variable. Then for all $t\in \RR^{>0}$,
            $\sup_{a\in \RR} P(-t/2<N-a\leq t/2)\leq \mathrm{TV}(N,N+t).$
\end{restatable}
\begin{proof}[Proof Sketch]
We consider a hypothesis test of $H_0: N$ versus $H_1: N+t$, using the rejection region $(a+t/2,\infty)$. The type I and type II errors of this test correspond to the tail probabilities $P(N-a> t/2)$ and $P(N-a\leq t/2)$; see Figure \ref{fig:antiLemma}. Combining this with the inequality $(\text{type I})+(\text{type II})\leq 1-\mathrm{TV}(N,N+t)$ gives the result.
\end{proof}

\begin{remark}
    Since total variation is a symmetric and translation-invariant relation, it follows that $\mathrm{TV}(N,N+t)=\mathrm{TV}(N+t,N)=\mathrm{TV}(N,N-t)$, which are alternative expressions for the right side of Lemma \ref{lem:anti}. 
\end{remark}

If a random variable $N$ is an additive noise mechanism at sensitivity 1 for $f$-DP, then $T(N,N+1)\geq f$. The following theorem establishes that in this case, we can express the upper bound $\mathrm{TV}(N,N+t)$ from Lemma \ref{lem:anti} in terms of evaluations of $f$. 

\begin{restatable}{theorem}{thmanti}
    [Anti-Concentration for Additive Noise]\label{thm:anti}
    Let $N$ be a random variable such that $T(N,N+1)\geq f$ for a symmetric nontrivial tradeoff function $f$. Then for $t\in \ZZ^{>0}$, 
    \begin{align*}
       \sup_{a\in \RR} P(-t/2<N-a\leq t/2)&\leq \begin{cases}
            1-2f^{\circ k}(c_f)&\text{ if }t=2k+1,\\
            1-2f^{\circ k}(1/2)&\text{ if } t=2k,
        \end{cases}
    \end{align*}
    where $k\in \ZZ^{\geq 0}$ and the notation $f^{\circ k}$ represents $f\circ \cdots \circ f$, where $f$ appears $k$ times.
\end{restatable}
\begin{proof}[Proof Sketch]
By Lemma \ref{lem:group2}, we have that $T(N,N+t)\geq f^{\circ t}$. Then we apply the formula $\mathrm{TV}(N,N+t)\leq 1-2 c_{f^{\circ t}}$ to the result of Lemma \ref{lem:anti}. Finally, Lemma \ref{lem:Cf} gives the more explicit formula for $c_{f^{\circ t}}$, which depends on the parity of $t$.  The dependence on the parity comes from the expression $c_{f^{\circ t}}=F_f(-t/2)$, where $F_f$ is the CND for $f$ from Proposition \ref{prop:CNDsynthetic} below, which has a different recursive expression depending on whether $t/2$ is an integer or half-integer.  Note that Lemmas \ref{lem:Cf} and \ref{lem:group2} appear in the appendix.
\end{proof}

Theorem \ref{thm:anti} establishes that at all half-integer values, an additive noise distribution can only be so concentrated around zero. This result provides a minimum on the magnitude of noise that must be added to achieve $f$-DP. Lemma \ref{lem:anti} and Theorem \ref{thm:anti} are similar to anti-concentration inequalities of L\'evy's concentration function \citep{krishnapur2016anti}, which are of the form $\sup_{a\in \RR} P(|X-a|\leq t)\leq g(t)$ for some function of $t$. Note that when $N$ is continuous, the results of Lemma \ref{lem:anti} and Theorem \ref{thm:anti} simplify to this form.

\begin{remark}
In Lemma \ref{lem:anti} and Theorem \ref{thm:anti}, we could have written the probabilities as $P(-t/2\leq N-a<t/2)$. In fact, a more general version of these anti-concentration inequalities is as follows: for any $c\in [0,1]$ and any $a\in \RR$,
\[c P(N-a=-t/2)+P(-t/2<N-a<t/2)+(1-c) P(N-a=t/2)\leq \mathrm{TV}(N,N+1),\] where the ``$c$'' corresponds to the randomized test $\phi(x) = I(x>a+t/2)+cI(x=a+t/2)$ of $H_0:N$ versus $H_1:N+t$.
\end{remark}

In the following sections, we will see that different canonical noise distributions match the bounds of Lemma \ref{lem:anti} and Theorem \ref{thm:anti}, indicating that they are near-optimal to achieve $f$-DP.

\section{Continuous Canonical Noise Distributions}\label{s:continuous}
While $f$-DP was introduced by \citet{dong2022gaussian}, the first method of constructing a privacy mechanism for an arbitrary $f$-DP guarantee was proposed by \citet{awan2023canonical}. \citet{awan2023canonical} proposed the concept of a \emph{canonical noise distribution} (CND), which captured the idea that the noise distribution satisfies $f$-DP and the privacy guarantee is tight. They also gave a construction for a CND for an arbitrary tradeoff function, and showed that CNDs are connected to certain optimal hypothesis tests.  
In this section, we review background on CNDs, show that they match the inequality of Theorem \ref{thm:anti}, prove that CNDs are sub-exponential, and establish optimality properties of log-concave CNDs.

\subsection{Background on Canonical Noise Distributions}\label{s:CNDback}
In this section, we review some of the key results of \citet{awan2023canonical}. A canonical noise distribution is an additive mechanism, which tightly matches the $f$-DP bound and has other desirable properties such as symmetry and a monotone likelihood ratio for testing $T(N,N+1)$. \citet{awan2023canonical} proposed the following definition to capture these desirable properties:

\begin{definition}[Canonical Noise Distribution: \citet{awan2023canonical}]\label{def:CND}
  Let $f$ be a symmetric tradeoff function, and let $N$ be a continuous random variable with cumulative distribution function (cdf) $F$. Then, $F$ is a \emph{canonical noise distribution} (CND) for $f$ if 
  \begin{enumerate}
      \item 
      For any $m\in [0,1]$, $T(N,N+m)\geq f$,
      \item $f(\alpha)=T(N,N+1)(\alpha)$ for all $\alpha \in (0,1)$,
      \item $T(N,N+1)(\alpha) = F(F^{-1}(\alpha)-1)$ for all $\alpha \in (0,1)$,
      \item $F(x) = 1-F(-x)$ for all $x\in \RR$; that is, $N$ is symmetric about zero.
  \end{enumerate}
\end{definition}

In Definition \ref{def:CND}, the four properties can be interpreted as follows: 1) adding $\Delta N$ to a statistic of sensitivity $\Delta$ satisfies  $f$-DP, 2) if $S(D)$ and $S(D')$ differ by exactly the sensitivity $\Delta$, then the tradeoff function $T(S(D)+\Delta N,S(D')+\Delta N)=f$ showing that the privacy guarantee is tight, 3) for  $S(D)$ and $S(D')$ that differ by exactly the sensitivity $\Delta$, an optimal rejection region is of the form $[x,\infty)$, which implies a monotone likelihood ratio property, and 4) since the privacy guarantee $f$-DP is symmetric, it is sensible to restrict attention to symmetric distributions.

The following recurrence is an important technical property of CNDs:
\begin{lemma}[\citet{awan2023canonical}]\label{lem:recurrence}
  Let $f$ be a symmetric nontrivial tradeoff function and let $F$ be a CND for $f$. Then $F(x)=1-f(1-F(x-1))$ when $F(x-1)>0$ and $F(x)=f(F(x+1))$ when $F(x+1)<1$. 
\end{lemma}
Lemma \ref{lem:recurrence} implies that if a CND is specified on an interval of length 1, then it is fully determined by the recurrence.  \citet{awan2023canonical} showed that the above recurrence relation can be used to construct a CND for any nontrivial symmetric tradeoff function, by starting with the linear function  $c_f(1/2-x)+(1-c_f)(x+1/2)$ on the interval $[-1/2,1/2]$. This linear function is chosen to ensure that the resulting distribution is symmetric about zero and has a continuous cdf. 

\begin{proposition}[CND construction: \citet{awan2023canonical}]\label{prop:CNDsynthetic}
  Let $f$ be a symmetric nontrivial tradeoff function. We define $F_f:\RR\rightarrow \RR$ as 
  \[ F_f(x) = \begin{cases}
  f(F_f(x+1))&x<-1/2,\\
  c_f(1/2-x) + (1-c_f)(x+1/2)&-1/2\leq x\leq 1/2,\\
  1-f(1-F_f(x-1))&x>1/2.\\
  \end{cases}\]
  Then $N\sim F_f$ is a canonical noise distribution for $f$.
\end{proposition}

Proposition \ref{prop:CNDsynthetic} demonstrates that CNDs exist for any symmetric nontrivial tradeoff function, and gives an explicit construction. Furthermore, \citet{awan2023canonical} give an algorithm to sample from the constructed CND, enabling it to be employed in practice,  which we have implemented in R code for this paper.


\subsection{Near-Optimality of CNDs}\label{s:nearCND}
While \citet{awan2023canonical} demonstrated that CNDs are essential to constructing optimal hypothesis testing procedures, they did not give any direct results to show that CNDs add a minimal amount of noise. In this section, we show the concentration of a CND about zero is equal to the upper bound given in Theorem \ref{thm:anti} at all $t\in \ZZ^{\geq 0}$. This implies that at half-integer values, the CND is more concentrated than any alternative additive noise satisfying $f$-DP. We also show that the mass of any other $f$-DP noise distribution is no more than a radius of 1/2 closer to its center than a CND.

Lemma \ref{lem:CNDconcentration} gives an explicit formula for the concentration of a CND, which we see matches the bound in Theorem \ref{thm:anti}.

\begin{restatable}{lemma}{lemCNDconcentration}
    \label{lem:CNDconcentration}
    Let $f$ be a symmetric, nontrivial tradeoff function, and let $N\sim F_N$ be a CND for $f$. Then for all $t\in \ZZ^{\geq 0}$,
    \[P(|N|\leq t/2)= \begin{cases} 1-2f^{\circ k}(c_f)&\text{ if }t=2k+1,\\
    1-2f^{\circ k}(1/2)&\text{ if }t=2k,\end{cases}\]
    where $k\in \ZZ^{\geq 0}$
\end{restatable}

Combining Lemma \ref{lem:CNDconcentration} with Theorem \ref{thm:anti}, we get the following corollary, indicating that at half-integer values, CNDs are more concentrated than any alternative $f$-DP noise mechanism. 
\begin{restatable}{corollary}{corCNDoptimal}\label{cor:CNDoptimal}
    Let $f$ be a symmetric, nontrivial tradeoff function. Let $N\sim F$ be a CND for $f$, and let $N'$ be a random variable such that $T(N',N'+1)\geq f$. Then for every $t\in \ZZ^{\geq 0}$,
    \[P(|N|\leq t/2)\geq \sup_{a\in \RR} P(-t/2<N'-a\leq t/2).\]
\end{restatable}

It is important to acknowledge the limitations of Corollary \ref{cor:CNDoptimal}, as the inequality need not hold at values which are not half-integers. See Example \ref{ex:tulapLaplace} below:

\begin{figure}
    \centering
    \includegraphics[width=.48\linewidth]{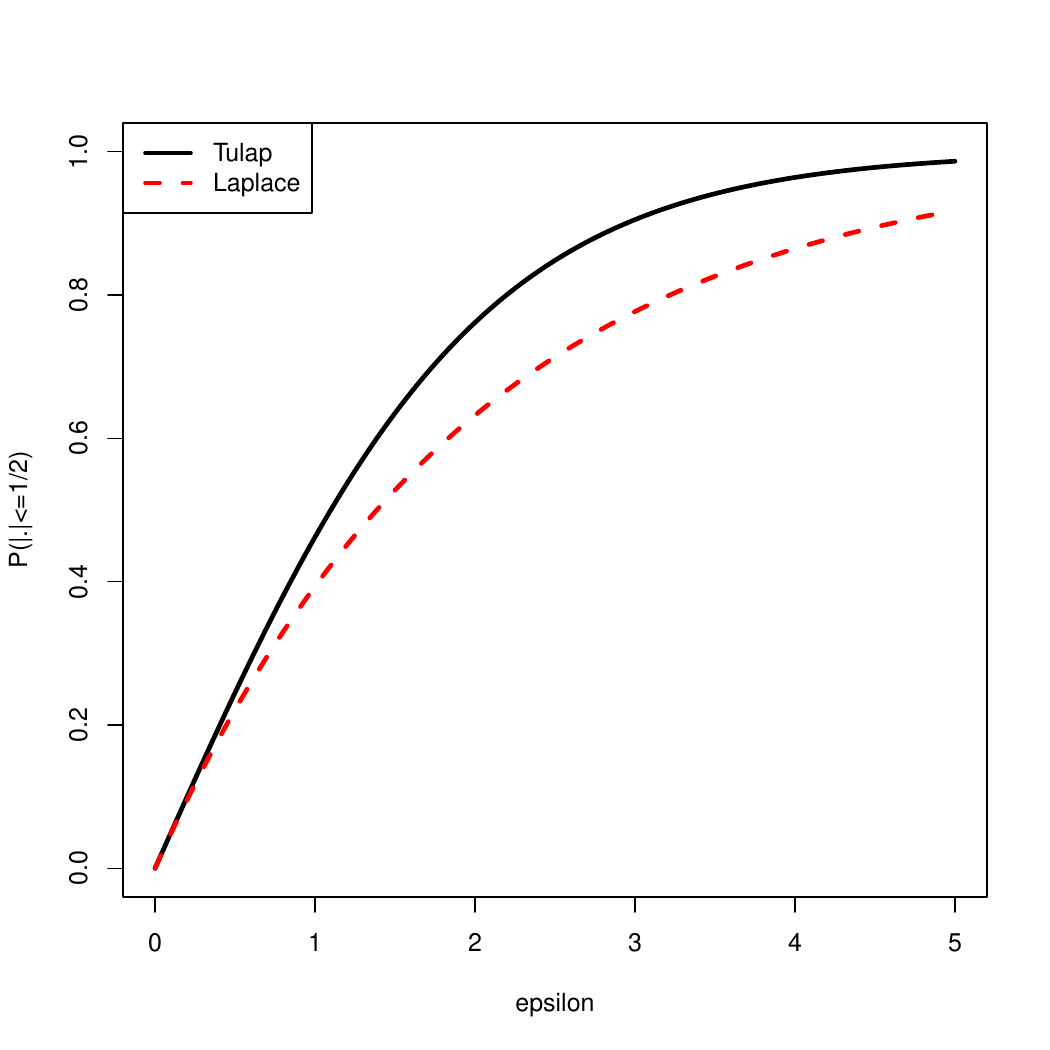}
        \includegraphics[width=.48\linewidth]{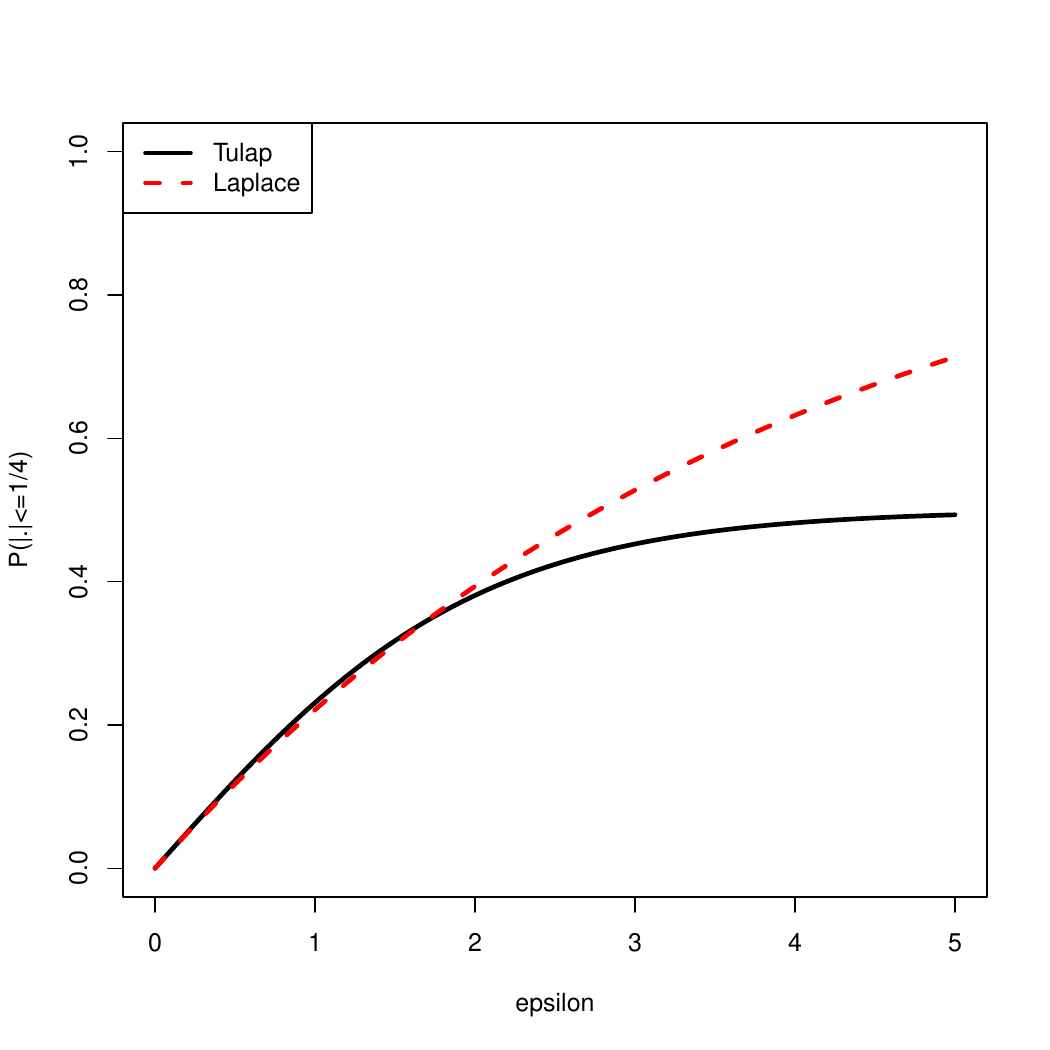}
    \caption{Comparison of central probabilities for Tulap and Laplace distributions, as $\epsilon$ varies. Left: $P(|\cdot|\leq 1/2)$, Right: $P(|\cdot|\leq 1/4)$. See Example \ref{ex:tulapLaplace} for details. }
    \label{fig:tulapLaplace}
\end{figure}

\begin{example}\label{ex:tulapLaplace}
    Recall from \citet{awan2022log} that $\mathrm{Tulap}(0,\exp(-\ep),0)$ is the unique CND for $(\ep,0)$-DP,  which can be constructed using Proposition \ref{prop:CNDsynthetic} \citep{awan2023canonical} ( the name is an abbreviation of truncated-uniform-Laplace, which references an alternative construction of the distribution). It is also common knowledge that $\mathrm{Laplace}(0,1/\ep)$ is another additive noise that can be used to satisfy $(\ep,0)$-DP. Let $N\sim \mathrm{Tulap}(0,\exp(-\ep),0)$ and let $L\sim \mathrm{Laplace}(0,1/\ep)$. We see in the left plot of Figure \ref{fig:tulapLaplace} that $P(|N|\leq 1/2)\geq P(|L|\leq 1/2)$ for all values of $\epsilon$. On the other hand, the right plot of Figure \ref{fig:tulapLaplace} shows that while $P(|N|\leq 1/4)\geq P(|L|\leq 1/4)$ for small values of $\epsilon$, when $\epsilon> 2$, $P(|N|\leq 1/4)\leq P(|L|\leq 1/4)$. This indicates that the result of Corollary  \ref{cor:CNDoptimal} cannot be extended for non-integer values of $t$. 
\end{example}

A consequence of Corollary \ref{cor:CNDoptimal} is that CNDs are nearly optimal compared to another $f$-DP noise distribution $N'$, in the sense that at half-integer values, the cdf of $|N|$ dominates the cdf of $|N'-a|$. Corollary \ref{cor:CNDratio} provides another perspective, demonstrating that $|N'-a|$ stochastically dominates $|N|-1/2$. This result can also be interpreted as saying that the mass of a competing $f$-DP noise distribution is no more than a radius of 1/2 closer to its center than a CND.

\begin{restatable}{corollary}{corCNDratio}\label{cor:CNDratio}
    Let $f$ be a symmetric, nontrivial tradeoff function. Let $N\sim F$ be a CND for $f$, and let $N'$ be a random variable such that $T(N',N'+1)\geq f$. Then for every $t\in \RR^{\geq 0}$,
    \[P(|N|\leq t+1/2)\geq \sup_{a\in \RR} P(-t<N'-a\leq t).\]
    If $N'$ is a continuous random variable, this simplifies to 
    $P(|N|\leq t+1/2)\geq \sup_{a\in \RR}P(|N'-a|\leq t).$
\end{restatable}

\begin{remark}
    While Lemma \ref{lem:CNDconcentration} establishes that the CND for a tradeoff function $f$ matches the anti-concentration bounds of Theorem \ref{thm:anti}, these are not the only distributions that can do so. For example, we see in Theorem \ref{thm:anti} that the bounds only depend on the values $f^{\circ k}(c_f)$ and $f^{\circ k}(1/2)$, which often do not fully determine the tradeoff function $f$. Thus, if $g\geq f$ satisfies $c_g=c_f$, $g^{\circ k}(c_f)=f^{\circ k}(c_f)$ and $g^{\circ k}(1/2)=f^{\circ k}(1/2)$ for all $k\geq 0$, then a CND for $g$ satisfies $f$-DP and also matches the anti-concentration bounds of Theorem \ref{thm:anti}.
\end{remark}

\subsection{CNDs are Sub-Exponential}\label{s:subExp}
In addition to comparing CNDs with other distributions, we can also directly analyze the tail behavior of CNDs. In the following result, we establish that all CNDs are sub-exponential, implying that their tail behavior has at slowest an exponential decay. As demonstrated in Example \ref{ex:cauchy} this result indicates that even if the target $f$-DP guarantee is designed in terms of a heavier tailed distribution (such as Cauchy), a CND only requires sub-exponential tails. 

Tail behavior of a noise mechanism is important to understand the dispersion, the probability of obtaining extreme values, and whether moments exist. With heavy-tailed distributions, such as the Cauchy distribution, it is possible that the moments do not exist. On the other hand, for sub-exponential random variables, all moments are guaranteed to be finite.

A random variable $X$ with mean $\mu$ is \emph{$(\sigma^2,\alpha)$-sub-exponential} if $\EE \exp(\lambda(X-\mu))\leq \exp(\lambda^2 \sigma^2/2)$ for all $|\lambda|<1/\alpha$. The bound on the moment generating function, implies an exponential decaying bound on the tails of the distribution by the Chernoff bound. In the following theorem, $\lfloor t\rfloor\defeq \max\{n\in \ZZ|n\leq t\}$ is the floor function.

\begin{restatable}{theorem}{thmsubexponential}\label{thm:sub-exponential}
    Let $f$ be a symmetric nontrivial tradeoff function and let $N\sim F$ be a CND for $f$. Set $\ep_f = \log\left((1-c_f)/c_f\right)$. Then the following hold:
\begin{enumerate}
    \item $P(|N|>t)\leq \exp(-\ep_f\lfloor t\rfloor)\leq \exp(-\ep_f (t-1))$ for $t\geq 0$,
    \item $\EE|N|^n\leq\ep_f^{-n}\exp(\ep_f)n!$ for $n\in \ZZ^{>0}$,
    \item $N$ is $(4\exp(\ep_f)/\ep_f^2,2/\ep_f)$-sub-exponential.
\end{enumerate}
\end{restatable}
\begin{proof}[Proof Sketch]
    Using part 3 of Lemma \ref{lem:Cfzero}, we upper bound $f$ with a piece-wise linear tradeoff function $f_{\ep_f,0}$. Using this along with the recurrence in Lemma \ref{lem:recurrence} we obtain the inequality in part 1. Part 2 is obtained by integrating the result of part 1, and part 3 is obtained from part 2 by using the Bernstein condition for sub-exponential random variables.
\end{proof}
An interesting consequence of Theorem \ref{thm:sub-exponential} is that there is no privacy-relevant reason to add noise with tails heavier than exponential. We explore this in terms of ``Cauchy-DP.''
\begin{example}[Cauchy-DP]\label{ex:cauchy}
    For $m\geq 0$, let 
    $C_m=T(\mathrm{Cauchy}(0,1),\mathrm{Cauchy}(m,1))$. Intuition  suggests that $C_m$-DP is a stronger notion of DP than $(\ep,0)$-DP, since Cauchy noise has heavier tails than Laplace or Tulap noise. However in this example, we establish that $C_m$-DP is equivalent to $(\ep,0)$-DP, up to a change of variables. 
    
    \citet[Example 6]{reimherr2019elliptical} previously proved that adding Cauchy noise to a statistic of sensitivity $m$ satisfies $(\ep_L(m),0)$-DP, where 
    \[\ep_L(m)=\log\left(\frac{4+(m+\sqrt{m^2+4})^2}{4+(m-\sqrt{m^2+4})^2}\right)\]
    (while their result is stated for multivariate $t$-distributions with degrees of freedom $\nu>1$, the argument also works for the Cauchy distribution, which is a univariate $t$-distribution with 1 degree of freedom). This implies that $C_m\geq f_{\ep_L(m),0}$. This means that ``Cauchy-DP'' is a stronger notion of privacy than $(\ep,0)$-DP,  which agrees with our intuition.
    
    On the other hand, using part 3 of Lemma \ref{lem:Cfzero}, we have that $C_m\leq f_{\ep_U(m),0}$, where $\ep_U(m) = \log\left((1-c_m)/c_m\right)$. We know that $1-2c_m=\mathrm{TV}(\mathrm{Cauchy}(0,1),\mathrm{Cauchy}(m,1))$ $= (2/\pi) \arctan\left(\frac{m}{2}\right)$ \citep{nielsen2022f}, where $c_m$ is shorthand for $c_{C_m}$. It follows that 
    \[\ep_U(m) = \log \left(\frac{3\pi-\arctan(m/2)}{2\arctan(m/2)-\pi}\right).\]
    In Figure \ref{fig:CauchyDP}, we plot $C_1$ as well as $f_{\ep_U,0}$ and $f_{\ep_L,0}$. 
%
    Theorem \ref{thm:sub-exponential} tells us that any CND for $C_m$ will have sub-exponential tails, indicating that a CND can obtain the same privacy guarantee as the Cauchy distribution, but with less noise. 
    
    In the right plot of Figure \ref{fig:CauchyDP} we compare $P(|N_i|\leq t)$ where $N_2\sim \mathrm{Cauchy}(0,1)$ and $N_1$ is drawn from the CND for $C_1$ constructed in Proposition \ref{prop:CNDsynthetic}. Empirically, for $t\geq 1/2$, $P(|N_1|\geq t)\geq P(|N_2|\geq t)$ and the difference is at most $\approx.00425$ for $t\leq 1/2$. Thus, the CND would be preferred over the Cauchy distribution in most applications. To illustrate the difference in the tails, we simulated 100 random variables from both distributions\footnote{ The code for this paper, available at \url{https://github.com/JordanAwan/OptimizingNoiseForFDP}, gives an implementation to simulate from the CND constructed in Proposition \ref{prop:CNDsynthetic}.}. The maximum absolute value of the CND variables was 4.28 whereas the maximum absolute value of the Cauchy random variables was 118.43. We see that the subexponential tails of the CND significantly reduce the chances of observing extreme events.
\end{example}

\begin{figure}
    \centering
    \includegraphics[width=.48\linewidth]{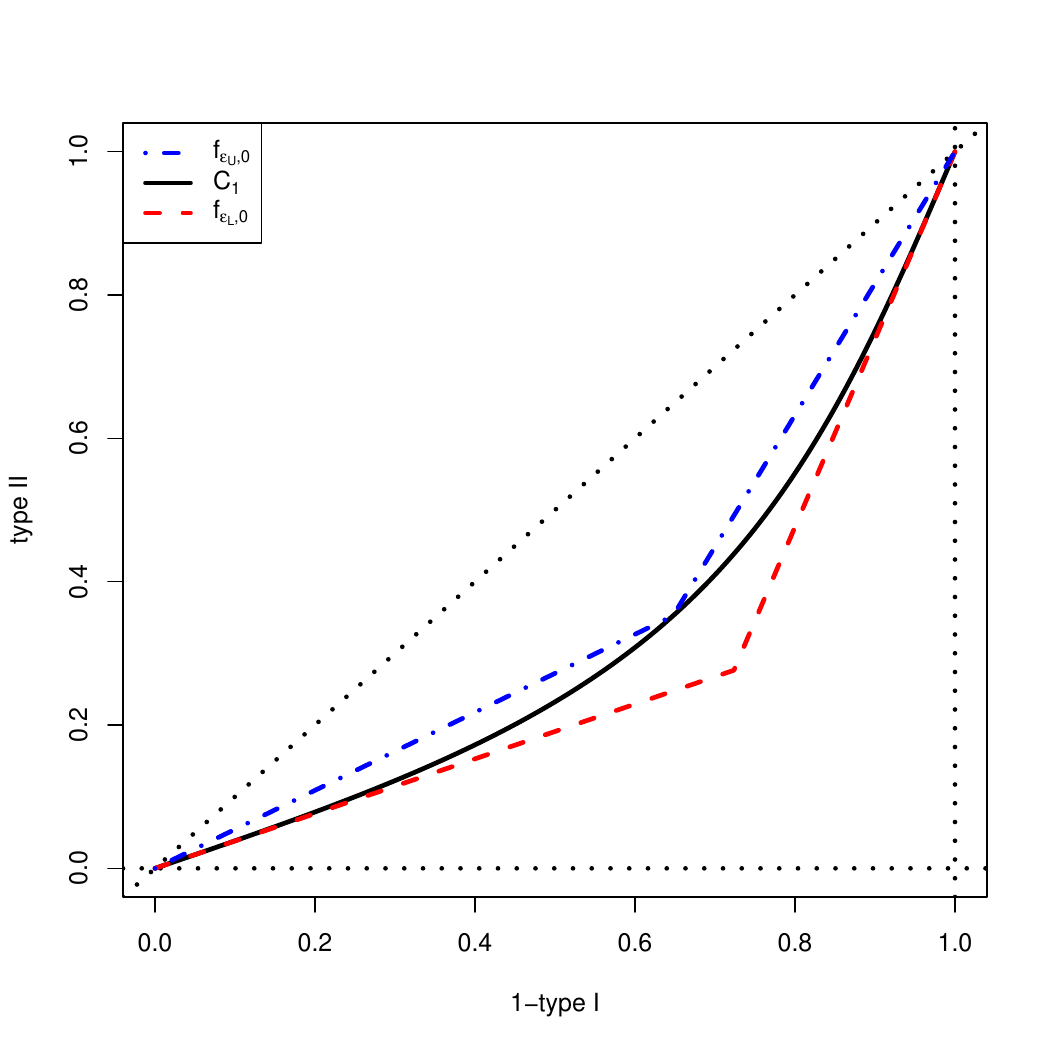}
    \includegraphics[width=.48\linewidth]{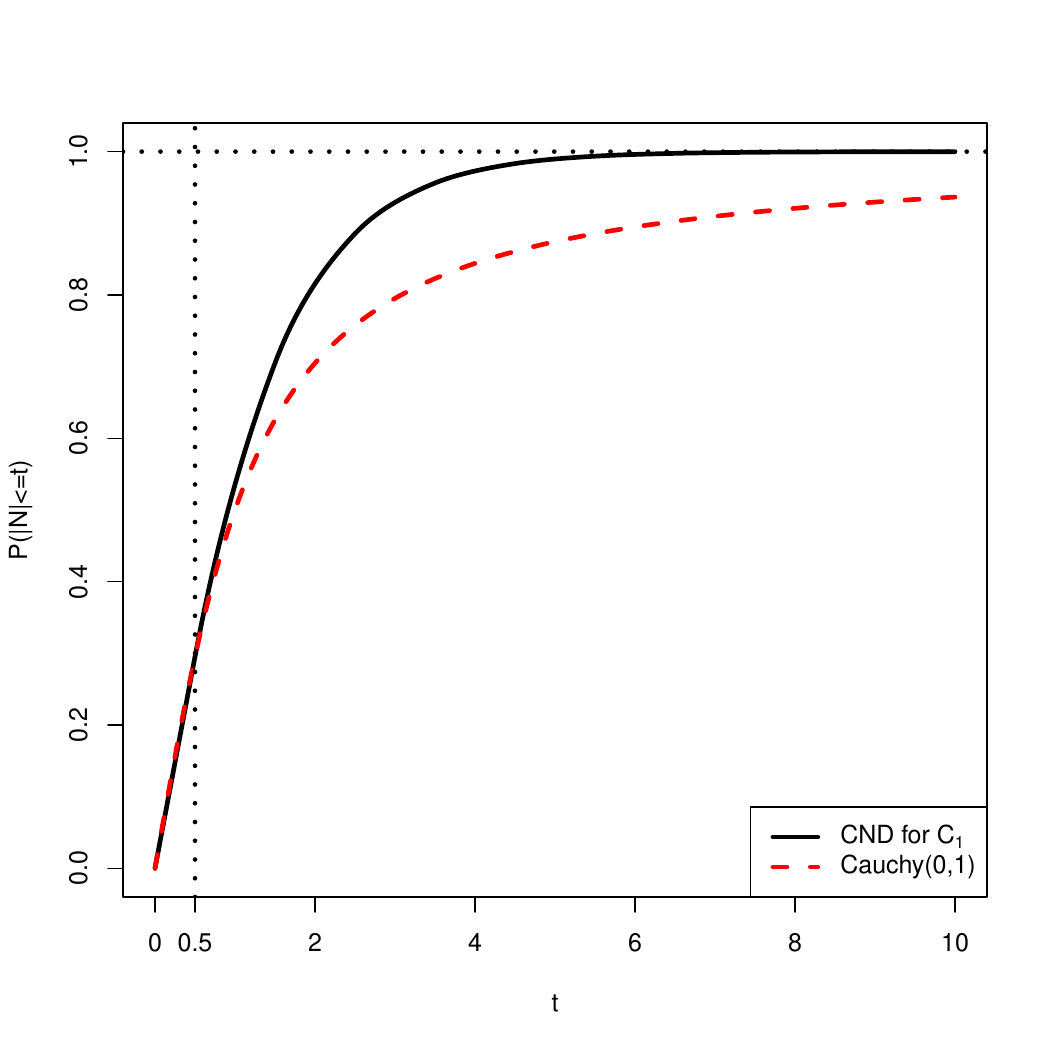}
    \caption{Left: Tradeoff function $C_1=T(\mathrm{Cauchy}(0,1),\mathrm{Cauchy}(1,1))$ as well as $f_{\ep_U,0}$ and $f_{\ep_L,0}$, from Example \ref{ex:cauchy}.  Right: Comparison of $P(|N|\leq t)$ where $N\sim \mathrm{Cauchy}(0,1)$ or $N$ is a CND for $C_1$. Vertical line is at $t=1/2$ and horizontal line is at 1.}
    \label{fig:CauchyDP}
\end{figure}

\subsection{Optimality of Log-Concave CNDs}\label{s:logCND}
Among the additive noise distributions log-concave distributions form an important subclass. Log-concave CNDs are also continuous CNDs and because of this they inherit all of the properties
of Sections \ref{s:nearCND} and \ref{s:subExp}. In this section, we give a new construction for the (unique) log-concave CND, when it exists, and we also show that log-concave CNDs are stochastically smallest among additive noises with the same privacy guarantee.

A continuous random variable is log-concave if its density can be expressed in the form $g(x)\propto \exp(C(x))$, where $C$ is a concave function. Log-concave distributions have many nice properties. In particular, a random variable $N$ with cdf $F$ satisfies $T(N,N+t)(\alpha)=F(F^{-1}(\alpha)-t)$ for every $t>0$ if and only if it is log-concave \citep{dong2022gaussian}. This is equivalent to stating that log-concave distributions have the monotone likelihood ratio property when testing between any two shifted versions of $N$.

 \citet{awan2023canonical} established conditions for the existence and construction of log-concave CNDs in terms of whether a tradeoff function is \emph{infinitely divisible}. A collection of tradeoff functions $\{f_t|t\geq0\}$ is infinitely divisible if it has the following properties: 1) $f_s\circ f_t=f_{s+t}$ for all $s,t\geq 0$, 2) $f_s$ is nontrivial for all $s>0$, and 3) $f_{s}\rightarrow f_0=\mathrm{Id}$ as $s\downarrow 0$, where $\mathrm{Id}(\alpha)=\alpha$. Infinite divisibility formalizes whether the tradeoff functions in the collection can be achieved by group privacy. \citet{awan2022log} showed that a symmetric nontrivial tradeoff function has a log-concave CND if and only if it belongs to an infinitely divisible collection of tradeoff functions. Furthermore, \citet{awan2022log} gave a construction for the unique log-concave CND of an infinitely divisible collection of tradeoff functions, which was expressed as a limit of CNDs constructed via Proposition \ref{prop:CNDsynthetic}. In the following lemma, we give a more explicit construction of the same log-concave distribution. Interestingly, the log-concave CND only depends on a single value from each of the $f_t$'s: $f_t(1/2)=c_{f_{2t}}$.  Lemma \ref{lem:logConcaveConstruction} also offers another way of understanding why the log-concave CND for $f_1$ is unique.

\begin{restatable}{lemma}{lemlogConcaveConstruction}\label{lem:logConcaveConstruction}
    Let $\{f_t|t\geq 0\}$ be an infinitely divisible collection of tradeoff functions, and let $F$ be the log-concave CND for $f_1$. Then, for $t\in \RR^{>0}$,
    $F(-t)=f_t(1/2)=c_{f_{2t}}.$ 
\end{restatable}

Most importantly, Lemma \ref{lem:logConcaveConstruction} can be used to construct the log-concave CND $F$ from the family of tradeoff functions, but it can also be used to compute the value $c_{f_t}$, when $f_t$ is log-concave:

\begin{example}
    [Laplace $c_f$]
    Let $f_t(\alpha)=T(\mathrm{Laplace}(0,1),\mathrm{Laplace}(t,1))$, which is  infinitely divisible, with log-concave CND $\mathrm{Laplace}(0,1)$ for $f_1$. Then, we can calculate 
    $c_{f_t}=F(-t/2)=(1/2)\exp(-t/2),$ 
    which is simply an evaluation of the Laplace cdf. 
\end{example}


In Theorem \ref{thm:stochLog}, we establish that the log-concave CND $N$ is stochastically smallest, compared to any other noise distribution which satisfies $T(N',N'+t)\geq f_t$ for all $t\geq 0$. 


\begin{restatable}{theorem}{thmstochLog}\label{thm:stochLog}
    Let $\{f_t|t\geq 0\}$ be an infinitely divisible collection of tradeoff functions, and let $F$ be the log-concave CND for $f_1$. Let $N\sim F$ and let $N'$ be any random variable such that $T(N',N'+t)\geq f_t$ for all $t\geq0$. Then $|N'-a|$ stochastically dominates $|N|$ for all $a\in \RR$. It follows that 
    \begin{itemize}
        \item $P(|N|\leq t)\geq P(|N'-a|\leq t)$ for all $a\in \RR$ and all $t\in \RR^{\geq 0}$,
        \item For any non-decreasing function $\phi$, $\EE_N\phi(|N|)\leq \EE_{N'}\phi(|N'-a|)$ for all $a\in \RR$,
        \item There exists a non-negative, random variable $W$ such that $|N'|\overset d =|N|+W$. If $N'$ is symmetric about zero, then $N'\overset d = N+\mathrm{sign}(N)\cdot W$. 
    \end{itemize}
\end{restatable}

While the condition $T(N',N'+t)\geq f_t$ for all $t\geq 0$ is not required to achieve $f$-DP (it is only required that $T(N',N'+t)\geq f$ for $t\in [-1,1]$), this stronger condition is related to \emph{individual privacy accounting} \citep{rogers2016privacy,feldman2021individual,koskela2023individual}, which tracks each individual's privacy budget along several sequential compositions, such as in DP stochastic gradient descent. By using infinitely divisible tradeoff functions along with the condition $T(N',N'+t)\geq f_t$ for all $t\geq 0$, we can use $f_t$ as a lower bound for privacy accounting within the $\{f_t|t\geq 0\}$ family. Theorem \ref{thm:stochLog} shows that among these types of noise distributions, the log-concave CND is optimal. 

Stochastic dominance is a very strong ordering of random variables. While most prior works focus on a specific objective criterion to optimize (e.g., mean absolute and mean squared error \citep{geng2015approx}, optimal hypothesis testing \citet{awan2018differentially}, Wasserstein distance \citep{qin2022differential2}), stochastic dominance implies that the mechanism optimizes all symmetric objectives centered at the non-private value, which are non-decreasing away from the center.

\begin{example}
    Examples of log-concave distributions include Gaussian, Laplace, Logistic, Uniform, and Beta, as well as truncated versions of these distributions. Of these, the Gaussian distribution is the log-concave CND for GDP, Laplace is the log-concave CND for Laplace-DP \citep{awan2022log}, and uniform is the log-concave CND for $(0,\delta)$-DP \citep{awan2022log}. \citet{dong2020gaussian} showed that the logistic tradeoff function is a lower bound on the privacy guarantee for exponential mechanisms \citep{mcsherry2007mechanism}. 
\end{example}

\section{Discrete Canonical Noise Distributions}\label{s:discrete}
In Section \ref{s:continuous}, the canonical noise distributions we considered were all continuous, as were those considered in the previous literature \citep{awan2023canonical,awan2022log}. However, there are also important use-cases where integer-valued noise is preferable. For example, the US Census Bureau used the discrete Gaussian mechanism \citep{canonne2020discrete} to privatize the 2020 US Census products. Other papers have advocated for discrete noise such as the Skellam distribution \citep{agarwal2021skellam}, binomial distribution \citep{dwork2006our,agarwal2018cpsgd}, and discrete Laplace/geometric mechanism \citep{ghosh2012universally}. Even when one may theoretically prefer a continuous noise distribution, there may still be a benefit of using a discrete noise distribution, to ensure that an implementation on finite computers has guaranteed privacy properties (e.g., floating point calculations are vulnerable to privacy attacks \citep{mironov2012significance}).

In this section, we propose a definition for a ``discrete CND,'' which captures the idea that 1) it is an integer-valued distribution, 2) for integer-valued statistics with finite sensitivity, it can be used to achieve $f$-DP,  3) the privacy guarantee is ``tight,'' and 4) it satisfies a monotone-likelihood ratio property analogous to property 3 in Definition \ref{def:CND}.

\begin{definition}\label{def:dCND}
Let $\Delta \in \ZZ^{>0}$, let $f$ be a symmetric tradeoff function, and let $N$ be a random variable with cdf $F$. Then $F$ is a \emph{discrete CND} for $f$ at sensitivity $\Delta$ if 
\begin{enumerate}
    \item  $T(N,N+t)\geq f$, for all $t\in \{-\Delta, \ldots, \Delta\}$,
    \item $f(F(t+\Delta))=F(t)$ for all $t \in \ZZ$, such that $F(t+\Delta)<1$,
    \item $N$ takes values in $\ZZ$, and is symmetric about zero.
\end{enumerate}
\end{definition}

Property 1 of Definition \ref{def:dCND} implies that for an integer-valued statistic $S(D)$, with sensitivity $\Delta$, $T(S(D)+N,S(D')+N)\geq f$ for any adjacent databases $D$, and $D'$.
Property 2 implies a monotone likelihood ratio property for $N$, and that the tradeoff function for the discrete CND is ``tight" for $f$. Essentially using the rejection set  $(t,\infty)$, where $t\in \ZZ$, the type I error for $T(N-\Delta, N)$ is $1-F(t+\Delta)$ and type II error is $F(t)$. 

Note that in property 2) of Definition \ref{def:dCND}, the tradeoff function $T(N,N+\Delta)$ only matches $f$ at rejection regions of the form $(t,\infty)$, where $t\in \ZZ$. In general, this means that the tradeoff function $T(N,N+\Delta)$ agrees with $f$ at the ``sharp points,'' but may be greater than $f$ at other values of $\alpha$. See Figure \ref{fig:discreteFDP} for an illustrative example. 

\begin{figure}
    \centering
    \includegraphics[width=.5\linewidth]{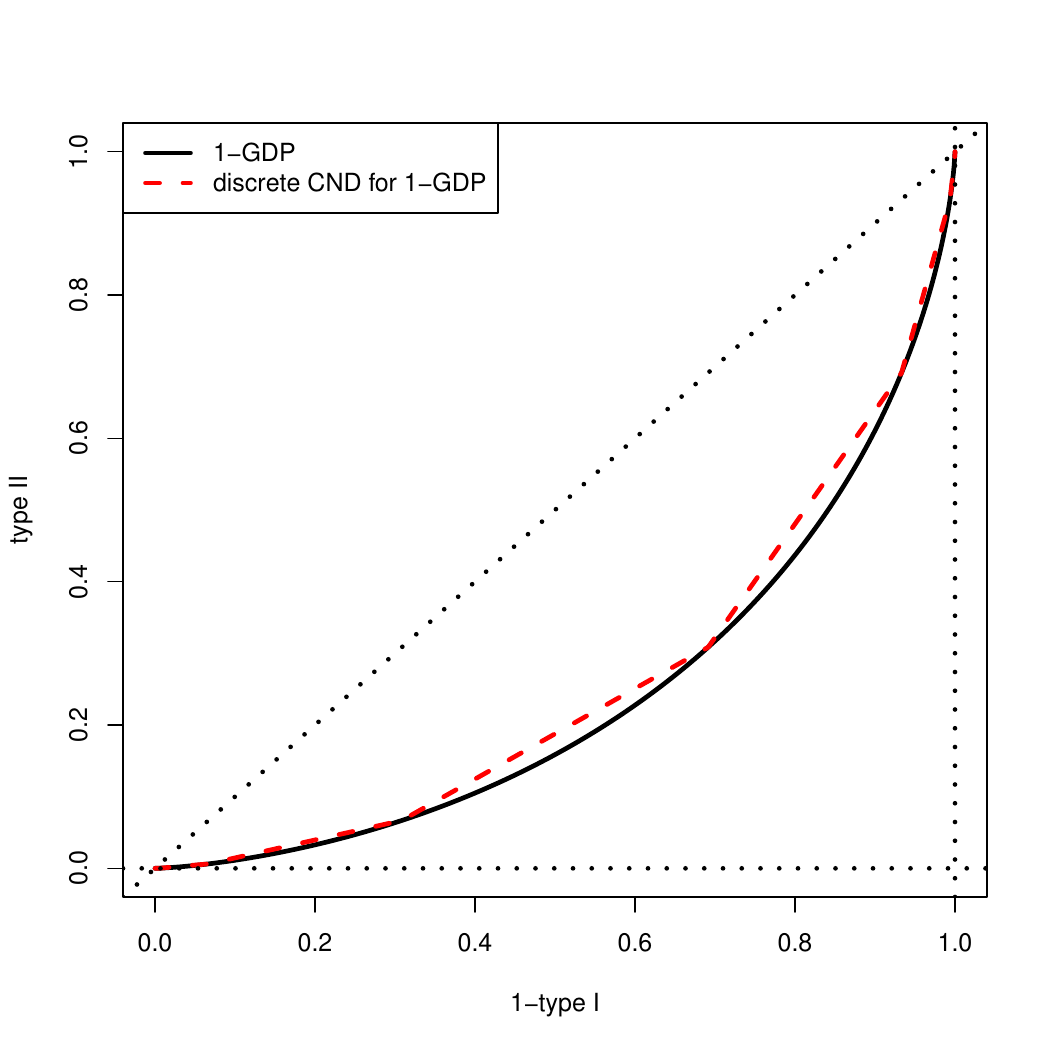}
    \caption{Tradeoff function of $1$-GDP as well as the tradeoff function $T(N,N+1)$, where $N$ is the unique discrete CND for $1$-GDP at sensitivity 1.}
    \label{fig:discreteFDP}
\end{figure}

\begin{remark}
    One key difference between Definition \ref{def:dCND} for discrete CNDs and Definition \ref{def:CND} for continuous CNDs is the dependence on the sensitivity. For continuous CNDs, the sensitivity was not included in the definition, as the noise can simply be scaled up or down according to the sensitivity. However, in the discrete case, scaling affects the support of the distribution. For example, if $N$ takes values on the integers and we try to scale it by 2 to adjust for statistics of sensitivity 2, then $2N$ takes values on the even numbers. In this case, when shifting by a smaller amount, we fail to guarantee $f$-DP:  $T(2N,2N+1)=0$. Instead, the discrete CND must be designed for a specific sensitivity, as indicated in Definition \ref{def:dCND}.
\end{remark}

\subsection{Existence, Construction, and Uniqueness of Discrete CNDs}\label{s:discreteBasics}
In Proposition \ref{prop:round}, we demonstrate that a discrete CND exists for any nontrivial symmetric tradeoff function and any sensitivity, and that we can construct a discrete CND by rounding a re-scaled continuous CND. In the case of sensitivity 1, we show in Proposition \ref{prop:unique} that the discrete CND is \emph{unique}.

In the following proposition, the round function is formally defined as $\mathrm{round}(t)=\lfloor t+1/2\rfloor$. 
\begin{restatable}[Existence and Construction via Rounding]{proposition}{propround}\label{prop:round}
Let $N_c$ be a (continuous) CND for $f$ with cdf $F$, and let $\Delta\in \ZZ^{>0}$.  Then $N=\mathrm{round}(\Delta N_c)$, which has cdf $F_N(t) = F_{N_c}(\lfloor t+1/2\rfloor/\Delta)$, is a discrete CND for $f$ at sensitivity $\Delta$. 
\end{restatable}

In particular, combining Proposition \ref{prop:round} with Proposition \ref{prop:CNDsynthetic}, we have an explicit construction for a discrete CND, and by using Proposition F.6 from \citet{awan2023canonical} (reprinted as Lemma \ref{lem:F6} in the appendix), we also have a sampling algorithm for this discrete CND. 

\begin{example}
    [Staircase Distribution]
    In the case of $(\ep,\de)$-DP, the discrete CND constructed by rounding the constructed CND of Proposition \ref{prop:CNDsynthetic} has a staircase shape, which has been identified in other DP literature as a near-optimal distribution \citep{geng2015optimal,qin2022differential}. The pmf of this distribution is illustrated in Figure \ref{fig:staircase}. 
\end{example}

\begin{example}
    [Discrete CNDs Not by Rounding]
    While Proposition \ref{prop:round} tells us that we can construct a discrete CND by rounding a continuous CND, there are also discrete CNDs that do not arise in this manner. Recall that for $(\ep,0)$-DP, the unique (continuous) CND is the Tulap distribution \citep{awan2022log}. It follows that for each $\ep$ and $\De$, there is a unique discrete CND for $(\ep,0)$-DP generated by rounding the continuous CND. However, in the case of $\ep=1$ and $\Delta=2$, we demonstrate that there is an infinite family of discrete CNDs:

    In the case of $\Delta=2$, constraints 2 and 3 of Definition \ref{def:dCND} leave only one degree of freedom: the choice of $F(0)$; this is because symmetry enforces that $F(-1)=1-F(0)$ and the recurrence of 2 fixes all other values of $F$. One can then verify that for 
    \[0.5938455 \approx \frac{2\exp(\ep)}{3\exp(\ep)+1}\leq F(0)\leq \frac{\exp(\ep)+1}{\exp(\ep)+3}\approx 0.6502446,\]
    one obtains a discrete CND for $(1,0)$-DP. Since there is more than one discrete CND for $(1,0)$-DP at $\Delta=2$, and there is only one continuous CND for $(1,0)$-DP, it follows that not all discrete CNDs are a rounding of continuous CNDs.
\end{example}

Interestingly, for the special case of $\Delta=1$, Proposition \ref{prop:unique} below establishes that there is a \emph{unique} discrete CND for a nontrivial symmetric tradeoff function.

\begin{figure}
    \centering
    \includegraphics[width=.5\linewidth]{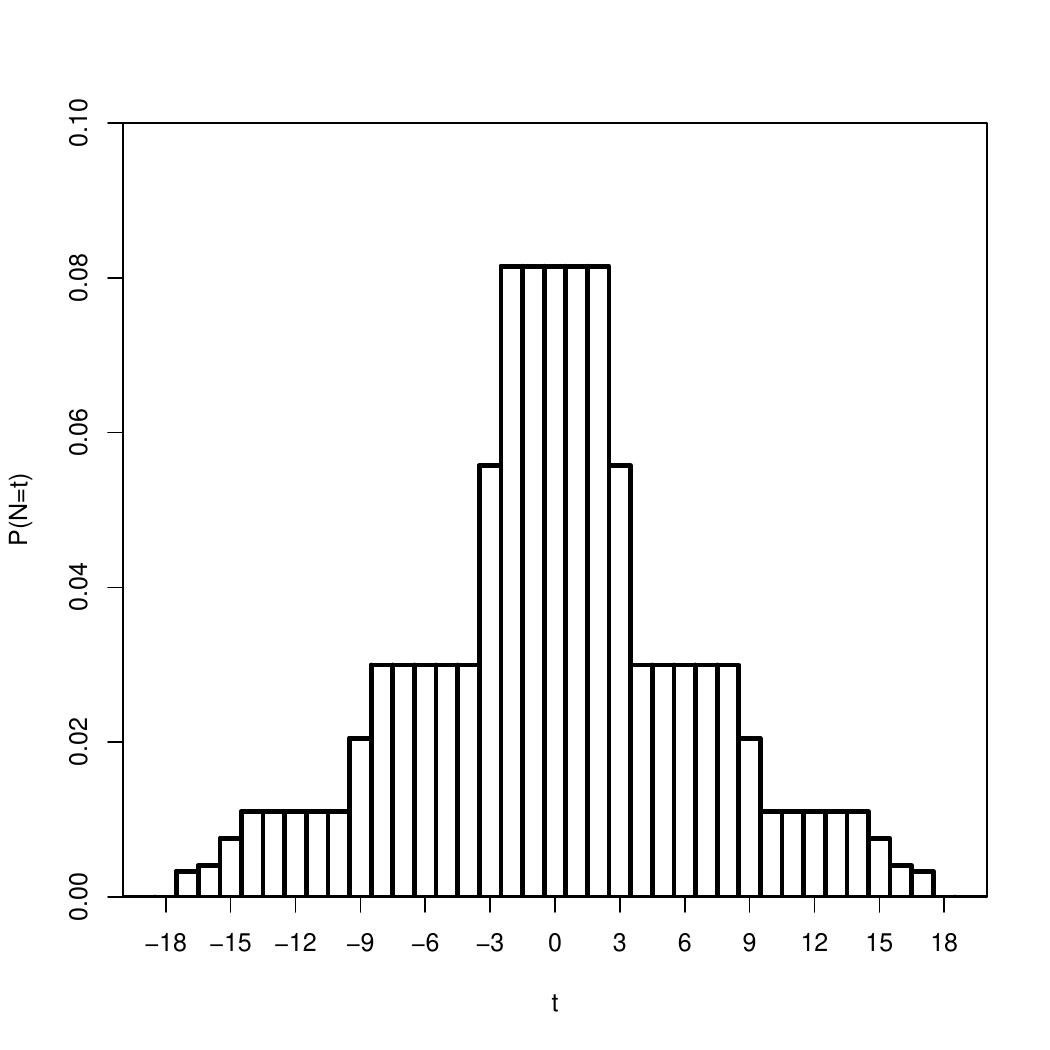}
    \caption{Discrete staircase distribution, which is a discrete CND for $(1,.05)$-DP at $\Delta=6$. }
    \label{fig:staircase}
\end{figure}

\begin{restatable}[Uniqueness]{proposition}{propunique}\label{prop:unique}
If $\Delta=1$, then there is a unique discrete CND for a symmetric nontrivial tradeoff function $f$. Furthermore, it can be realized by rounding any (continuous) CND for $f$, and has pmf
$P(N=x) = f^{\circ |x|}(1-c_f)-f^{\circ |x|}(c_f),\quad x\in \ZZ.$
\end{restatable}

This result follows from the observation that once $F(0)$ is chosen, the recurrence of property 2 of Definition \ref{def:dCND} fully determines the cdf, and the symmetry via property 3 of Definition \ref{def:dCND} removes the freedom in the choice of $F(0)$. 
Note further that combining Propositions \ref{prop:round} and \ref{prop:unique} tells us that rounding any CND (including the one constructed in Proposition \ref{prop:CNDsynthetic}) will result in the unique discrete CND at sensitivity 1. 

\begin{example}[Discrete Laplace]\label{ex:dLap}
The unique discrete CND for $f_{\ep,0}$, at sensitivity 1 is the discrete Laplace / geometric mechanism, with pmf
\[P(N=x) = \frac{\exp(\ep)-1}{\exp(\ep)+1} \exp(-\ep |x|),\quad x\in \ZZ.\]
This distribution is well-known in the differential privacy literature, and has been identified as the optimal noise-adding mechanism in pure differential privacy \citep{ghosh2012universally} for count statistics (which have sensitivity 1). 
\end{example}

\begin{example}
    [Rounded Gaussian versus Discrete Gaussian]\label{ex:dGauss1}
    Proposition \ref{prop:round} tells us that rounding a Gaussian distribution results in a discrete CND for Gaussian-DP (since Gaussian noise is a CND). Furthermore, by Proposition \ref{prop:unique}, the rounded Gaussian is the unique discrete CND at sensitivity 1. The rounded Gaussian has been referred to as the discrete Normal distribution, and has been independently studied in the statistics literature \citep{roy2003discrete}; it has pmf $P(N=x)=\Phi((x+1/2)/\sigma)-\Phi((x-1/2)/\sigma)$. On the other hand, the Discrete Gaussian distribution  \citep{canonne2020discrete} (also sometimes called the discrete Normal distribution \citep{kemp1997characterizations}) has gained popularity as a privacy mechanism, and has been employed in the 2020 Decennial Census. The Discrete Gaussian with parameter $\mu=0$ and scale parameter $\sigma$ has pmf 
    \[P(N=x) = \frac{\exp(-x^2/(2\sigma^2))}{\vartheta_3(0,\exp(-1/(2\sigma^2)))},\]
    where $\vartheta_3(0,q) = \sum_{k=-\infty}^\infty q^{k^2}$ is a Jacobi theta function \citep{szablowski2001discrete}.
    
    In this example, we demonstrate that while the Discrete Gaussian distribution with parameter $\sigma^2=2/\rho$ satisfies $\rho$-zCDP for integer-valued statistics of sensitivity 1 (the same zCDP guarantee as the continuous Gaussian mechanism with the same scale parameter, \citealp{canonne2020discrete}), the discrete Gaussian with $\sigma=1/\mu$ does not satisfy $\mu$-GDP. For example, with $\mu=\sigma=1$, and setting $f_{DG} = T(N,N+1)$, where $N$ is a discrete Gaussian with $\sigma=1$, we have by the monotone likelihood ratio property of the discrete Gaussian that \[c_{f_{DG}}=(1/2)(1-P(N=0))=(1/2)(1-[\vartheta_3(0,\exp(-1/2(\sigma^2)))]^{-1})\approx 0.301,\]
    which is smaller than $c_{G_1}=\Phi(-1/2)\approx 0.309$, which implies that $f_{DG}\not\geq G_1$; see  Figure \ref{fig:discreteGaussian}.
    \end{example}

    \begin{figure}
        \centering
        \includegraphics[width=.5\linewidth]{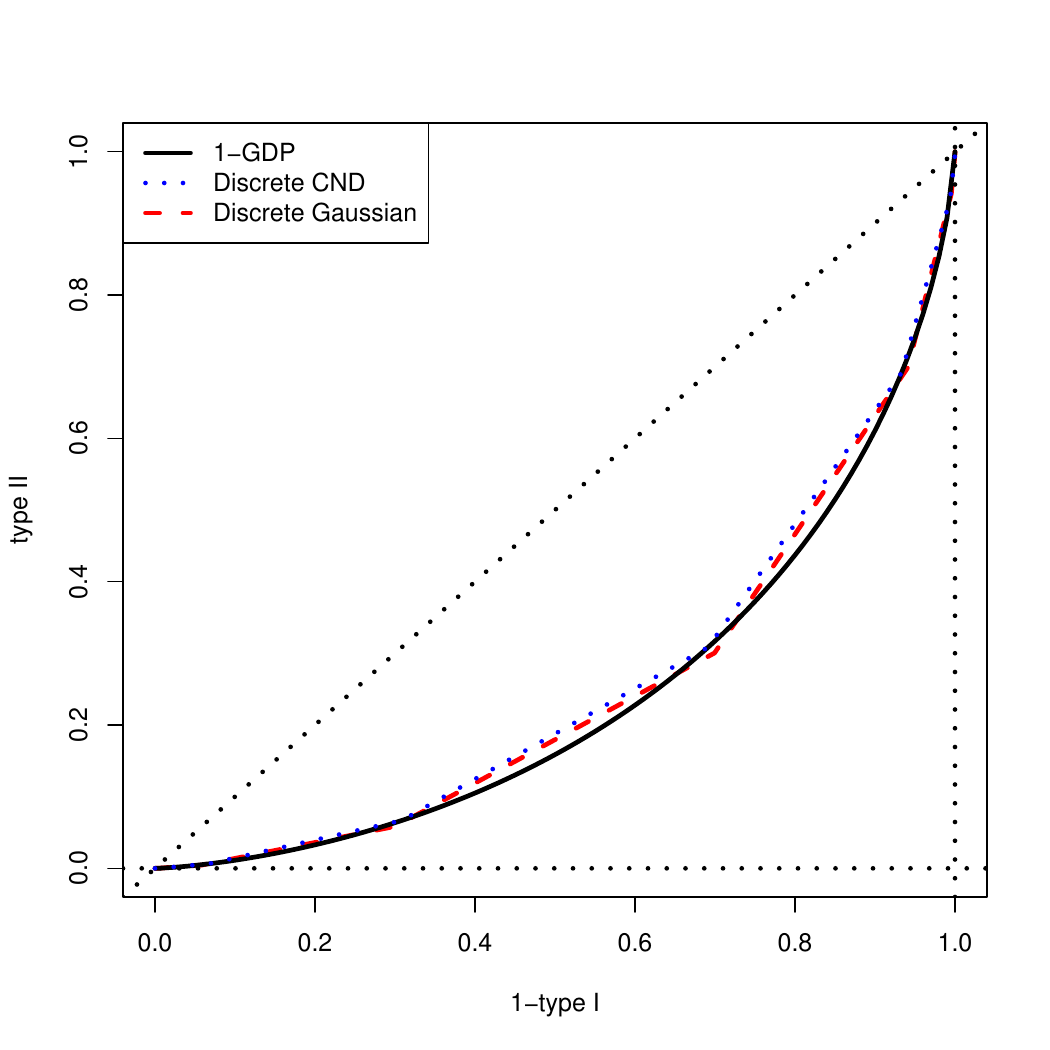}
        \caption{Tradeoff function for $T(N,N+1)$, where $N$ is a discrete Gaussian distribution with $\mu=0$ and $\sigma=1$, against the tradeoff function $G_1$, as well as the discrete CND for $G_1$ at sensitivity 1. Note that $T(N,N+1)$ is below $G_1$ at some values, violating the GDP guarantee.}
        \label{fig:discreteGaussian}
    \end{figure}

\begin{remark}
While a tradeoff function has a unique discrete CND for sensitivity 1, there may be multiple tradeoff functions with the same discrete CND for sensitivity 1. In particular, we see from Proposition \ref{prop:round} that the discrete CND only depends on $\{f^{\circ n}(1-c_f)\}_{n=0}^\infty$ or equivalently $\{f^{\circ n}(c_f)\}_{n=-1}^\infty$. So, if two tradeoff functions agree on these points, then they will have the same discrete CND for sensitivity 1. 
\end{remark}

\subsection{Optimality of Discrete CND at Sensitivity 1}\label{s:discreteOpt}
In this section, we establish a special case of Theorem \ref{thm:anti} for integer-valued noise, and show that the unique discrete CND at $\Delta=1$ matches the bound. Using this result, we establish that the discrete CND at $\Delta=1$ is stochastically smaller than any other integer-centered discrete noise.

While Theorem \ref{thm:anti} offers an anti-concentration inequality for any additive noise, the inequality is not enforced at all values of the support. For integer-valued noises for statistics with sensitivity 1, Theorem \ref{thm:anti} can give an inequality at every integer.

\begin{restatable}{corollary}{corinteger}\label{cor:integer}
Let $N$ be an integer-valued random variable such that $T(N,N+1)\geq f$, where $f$ is a nontrivial symmetric tradeoff function. Then for all $t\in \ZZ^{\geq 0}$, 
\[\sup_{a\in \ZZ}P(|N-a|\leq t)\leq 1-2f^{\circ t}(c_f).\]
\end{restatable}

Similar to Lemma \ref{lem:CNDconcentration}, Lemma \ref{lem:sens1} shows that the unique discrete CND at sensitivity 1 matches the bound of Corollary \ref{cor:integer}.
\begin{restatable}{lemma}{lemsensOne}\label{lem:sens1}
    Let $N$ be the unique discrete CND for a symmetric nontrivial tradeoff function $f$ at sensitivity $\Delta=1$. Then for $t\in \ZZ^{\geq 0}$, 
    $P(|N|\leq t) = 1-2f^{\circ t}(c_f).$
\end{restatable}

It follows from Lemma \ref{lem:sens1} that the unique discrete CND is stochastically smaller than any alternative DP noise, a similar result to Theorem \ref{thm:stochLog}.


\begin{restatable}{theorem}{thmstochDisc}[Stochastic Optimality of Discrete CND for Sensitivity 1]\label{thm:stochDisc}
Let $N$ be the unique discrete CND for a symmetric nontrivial tradeoff function $f$ at sensitivity $\Delta=1$, and let $N'$ be any integer-valued random variable such that $T(N',N'+1)\geq f$. Then, $|N'-a|$ stochastically dominates $|N|$ for all $a\in \ZZ$. This implies that,
\begin{itemize}
    \item $P(|N|\leq t)\geq P(|N'-a|\leq t)$ for all $a\in \ZZ$ and all $t\in \ZZ^{\geq 0}$,
    \item For any non-decreasing function $\phi$, $\EE_N\phi(|N|)\leq \EE_{N'}\phi(|N'-a|)$ for all $a\in \ZZ$,
    \item There exists a non-negative, integer-valued, random variable $W$ such that $|N'|\overset d =|N|+W$. If $N'$ is symmetric about zero, then $N'\overset d = N+\mathrm{sign}(N)\cdot W$. 
\end{itemize}
\end{restatable}

Theorem \ref{thm:stochDisc} is similar to the results of \citet{ghosh2012universally}. The key differences are 1) our results are for general $f$-DP instead of only $(\ep,0)$-DP, and 2) our results are phrased in terms of stochastic dominance rather than universal utility maximizers. In the case of $(\ep,0)$-DP both our Theorem \ref{thm:stochDisc} and the results of \citet{ghosh2012universally} suggest that the discrete Laplace mechanism, discussed in Example \ref{ex:dLap} is the optimal discrete noise for count statistics.

\begin{example}
Let $S$ be an integer-valued statistic with sensitivity 1. Suppose that we are considering two privacy mechanisms $M_1(D)=S(D)+N$ and $M_2(D)=S(D)+N'$, where $N$ is the discrete CND for $f$ at sensitivity 1 and $N'$ is another noise distribution satisfying $T(N',N'+1)\geq f$.
\begin{itemize}
\item Let $\psi$ be a symmetric loss function, which is increasing away from zero (e.g., $\psi(x)=|x|$ or $\psi(x)=x^2$). Then by Theorem \ref{thm:stochDisc}, $\EE \psi(M_1(D)-S(D))\leq \EE \psi(M_2(D)-S(D))$. This is a similar result to those in \citet{ghosh2012universally} and \citet{geng2015optimal}, indicating that the discrete CND at sensitivity 1 is a universal utility optimizer. 
\item If $\EE(N')=0$, then $\var(N)\leq \var(N')$, by taking $\phi(x)=x^2$ in Theorem \ref{thm:stochDisc}.
\end{itemize}
\end{example}

\begin{example}
    [Rounded Gaussian versus Discrete Gaussian, continued]\label{ex:dGauss2}
In Example \ref{ex:dGauss1}, we saw that the discrete Gaussian with $\sigma=1/\mu$ does not satisfy $\mu$-GDP for sensitivity 1 statistics. Even if there were a different scale parameter $\sigma(\mu)$, for which the discrete Gaussian does satisfy $\mu$-GDP, Theorem \ref{thm:stochDisc} tells us that the rounded Gaussian -- the discrete CND at sensitivity 1 for $\mu$-GDP -- will be stochastically smaller than the discrete Gaussian.
\end{example}

\begin{example}[Non-Integer-Valued Center]
If the noise $N'$ in Theorem \ref{thm:stochDisc} does not have an integer-valued mean, then the results of Theorem \ref{thm:stochDisc} no longer imply that  $\var(N)\leq \var(N')$. Recall that $U(-1,1)$ is a CND for $f_{0,1/2}$ \citep{awan2022log}. Then $N=\mathrm{round}(U)$ is the discrete CND for $f_{0,\delta}$ at sensitivity 1 and $N'=\lfloor U\rfloor$ is another integer-valued noise distribution which satisfies $T(N',N'+1)\geq f$ by postprocessing. We can calculate that $P(N=-1)=P(N=1)=1/4$, $P(N=0)=1/2$ and $P(N'=0)=P(N'=-1)=1/2$. From this, we see that $\var(N)=1/2>1/4=\var(N')$. 
 This counter-example demonstrates that the result of Theorem \ref{thm:stochDisc}
 cannot be extended to permit non-integer values of $a$, and that the variance may not be minimized by the discrete CND at sensitivity 1 if a competing discrete noise has a non-integer mean. 
\end{example}

\begin{example}[Non-Optimality for Continuous CNDs]
While Theorem \ref{thm:stochDisc} may seem intuitive, it is worth pointing out that the analogous result does not hold for continuous CNDs. Recall that $\mathrm{Tulap}(0,e^{-\ep},0)$ is the unique CND for $f_{\ep,0}$ \citep{awan2022log}. Recall further that $L\sim \mathrm{Laplace}(1/\ep)$ is an additive mechanism that can be used to satisfy $f_{\ep,0}$-DP. However, there is no stochastic dominance relationship between Tulap and Laplace. For example, at $\ep=5$, the variance of the Tulap is $\approx 0.097$, whereas the variance of the Laplace is $\approx 0.08$. On the other hand, $P(|N|\leq 1/2)$ is equal to $(\exp(\ep)-1)/(\exp(\ep)+1)\approx .987$ for Tulap and $1-\exp(-\ep/2) \approx .918$ for Laplace. Furthermore, \citet{geng2015optimal} derived the minimum variance and minimum mean absolute error additive (continuous) mechanisms for $(\ep,0)$-DP, which are not canonical noise distributions. So,  while a (continuous) CND optimizes the privacy budget, and matches the anti-concentration inequality (Lemma  \ref{lem:CNDconcentration}) at half-integer values, it does not necessarily optimize other objectives. 
\end{example}

\section{Discussion}\label{s:discussion}
Additive noise is a fundamental technique to achieve differential privacy, often either being employed by itself, or as part of a more complex privacy mechanism. Due to this, it is essential to understand the limits of what types of noise can be used to satisfy DP, and to optimize the noise. 

In this paper, we have explored the constraints of noise distributions under differential privacy, establishing upper bounds on the amount of mass that can be concentrated near the center of the distribution. We showed that canonical noise distributions (CNDs) match these bounds, which leads to their near-optimality. We also showed that log-concave CNDs are optimal compared to other noise distributions with the same privacy property. To address integer-valued statistics, we proposed a definition of a discrete CND which extended the original canonical noise distribution definition. We showed that discrete CNDs always exist and give an explicit construction. In the case that they are being added to a statistic of sensitivity 1, we showed that the discrete CND is unique and the smallest of any other integer-centered discrete noise. 

In addition to our theoretical contributions, the R code for this paper includes a general method to sample from the CND constructed in Proposition \ref{prop:CNDsynthetic} as well as functions to evaluate its cdf and quantile functions. These methods can enable researchers to implement and study CNDs in future works.\\


\noindent\textbf{Limitations and Future Work: } For continuous additive noise, the upper bounds on anti-concentration hold only at half-integer values. It would be worth exploring whether other bounds can be derived at other values. Alternatively, it would be interesting to investigate specific objective criteria, such as mean squared error, or mean absolute error, which have not been studied in the $f$-DP framework.

In terms of discrete noise, our definition of a discrete CND imposed that the distribution be symmetric about zero. The optimality result of Theorem \ref{thm:stochDisc} also only compared against other noise distributions centered at an integer value. Future research could weaken the discrete CND definition to not require symmetry about zero. While introducing a bias, this would allow for a wider range of noise mechanisms to be considered. 

We also point out that the optimality of the discrete CND, Theorem \ref{thm:stochDisc}, only applies in the sensitivity 1 case. Near-optimality results similar to Corollaries \ref{cor:CNDoptimal} and \ref{cor:CNDratio} can likely be derived for the discrete CND when when sensitivity is greater than 1. 

While this paper focused on additive noise mechanisms,  our results can be used to optimize the performance of more complex mechanisms which use noise addition as an intermediate step (e.g., stochastic gradient descent \citep{abadi2016deep}, objective perturbation \citep{chaudhuri2011differentially,kifer2012private}). Future researchers may investigate these privacy mechanisms to determine whether optimizing the additive noise distributions results in improved performance of the complete privacy mechanism.

This paper focused on the central DP model, where a mechanism is applied to the whole dataset $D$ which is held by a single curator. As \citet{dong2022gaussian} discussed, there is a natural version of $f$-local DP, where a mechanism is applied to a single entry at a time; this offers a stronger privacy protection because the mechanism can be applied to each individual's data before sending the results to the curator. With this formulation, all of our results still apply in the local DP setting. We leave it to future researchers to investigate any particular nuances that arise in the local DP setting, such as the design of iterative mechanisms.\\

\noindent\textbf{Multivariate Extension: } The results of this paper are limited to univariate noise distributions. Future researchers may develop a multivariate extension to Lemma \ref{lem:anti}, and explore whether there are optimality properties of the multivariate CNDs proposed by \citet{awan2022log}.

 While our techniques can be applied to multivariate distributions by restricting  to 1-dimensional projections, as demonstrated below, this has limited utility. Let $N\in \RR^d$ be a continuous  $f$-DP additive noise distribution with respect to the norm $\lVert \cdot \rVert$; that is, $T(N,N+v)\geq f$ for all $\lVert v \rVert\leq 1$ \citep{awan2022log}. Then,
\begin{align}
    P(\lVert N\rVert\leq t/2)&\leq \inf_{\lVert v \rVert =1}P(|\mathrm{Proj}_v N|\leq t/2)\\
    &\leq \inf_{\lVert v \rVert= 1} \mathrm{TV}(\mathrm{Proj}_v N,\mathrm{Proj}_v N+t)\label{eq:1d}\\
    &\leq \inf_{\lVert v\rVert= 1} \mathrm{TV}(N,N+tv)\label{eq:dpi}\\
    &\leq \begin{cases}
        1-2f^{\circ k}(c_f)&\text{if $t=2k+1$}\\
        1-2f^{\circ k}(1/2)&\text{if $t=2k$},
    \end{cases}\label{eq:last}
\end{align}
where $\mathrm{Proj}_v$ is the orthogonal projection onto the subspace spanned by the vector $v$, \eqref{eq:1d} applies Lemma \ref{lem:anti}, \eqref{eq:dpi} applies the data processing inequality for total variation, and \eqref{eq:last} uses the same argument as in the proof of Theorem \ref{thm:anti}. While this result does give some indication of the concentration of the distribution of $N$, it is generally not tight as the bound does not involve the dimension. The following example illustrates this gap in the case of Gaussian noise.
\begin{example}
The distribution $N(0,I)$ is a multivariate CND with respect to $\lVert \cdot\rVert_\infty$ for $G_{\sqrt d}$-DP \citep{awan2022log}. We can directly calculate
$P(\lVert N\rVert_\infty\leq t/2)=
[1-2\Phi(-t/2)]^d$.
 On the other hand, 
the bound \eqref{eq:dpi} simplifies to 
$P(\lVert N\rVert_\infty\leq t/2)\leq 1-2\Phi(-t/2),$
and we see that the bound is tight only in the case that $d=1$. 
\end{example}
We believe that the reason our method does not produce a tight bound in multivariate settings is because it is based on a two-point hypothesis test, similar in spirit to the Le Cam method for minimax lower bounds. It is possible that using a packing argument, similar to the Fano or Assouad methods, could produce a tighter bound that incorporates the dimension. 

\acks{We are thankful to Zhanyu Wang for feedback and discussion on an early draft of this manuscript. 
This work was supported in part by the National Science Foundation,  Grant No. SES-2150615 to Purdue University. }



\appendix

\section{Proofs and Technical Details}
 In this section, we include the proofs and technical details for the results in the paper.

\lemCfzero*
\begin{proof}
\begin{enumerate}
    \item Since $f(0)=0$, $f(1)\geq 0$, and $f$ is an increasing convex function, it follows that there is a unique solution $f(1-c_f)=c_f$. Since $f(\alpha)\leq \alpha$, it follows that $c_f\in [0,1/2]$. Furthermore, if $c_f=1/2$, then by convexity, it follows that $f(\alpha)=\alpha$ for all $\alpha$, implying that it is not nontrivial.
    \item Recall the equivalence between total variation and $f_{0,\delta}$: if $\delta$ is the smallest value such that $f_{0,\delta}\leq f$, then $\mathrm{TV}(P,Q)=\delta$. Recall that the $f_{0,\delta}$ curves have slope 1 and have value $1-\delta$ at $\alpha=1$. Since $f$ is symmetric, we see that the tightest $f_{0,\delta}$ curve supports $f$ at the point $(1-c_f,c_f)$. We calculate that the line with point $(1-c_f,c_f)$ with slope $1$ has value $2c_f$ at $\alpha=1$. We conclude that $\mathrm{TV}(P,Q)=1-2c_f$. 
    \item As discussed in the proof of part 2, $f_{0,1-2c_f}\leq f$.  
    For $f_{\ep,0}$, we call $c_\ep\defeq c_{f_{\ep,0}}$ as shorthand. We can calculate that $c_\ep = (1+e^\ep)^{-1}$, which implies that $\ep = \log\left(\frac{1-c_\ep}{c_\ep}\right)$. 
    Note that $f_{\ep_f,0}$ is the continuous, piece-wise linear tradeoff function with break points $(0,0)$, $(1-c_f,c_f)$, and $(1,1)$. Since $f$ is a tradeoff function, we have $f(0)=0=f_{\ep_f,0}(0)$, $f(1-c_f)=c_f=f_{\ep_f,0}$, and $f(1)\leq 1=f_{\ep_f,0}$. Thus, since $f$ is convex, it follows that $f(\alpha)\leq f_{\ep_f,0}(\alpha)$ for all $\alpha\in [0,1]$.
\end{enumerate}
\end{proof}

\begin{restatable}{lemma}{lemCNDcf}\label{lem:CNDcf}
Let $F$ be a CND for a symmetric nontrivial tradeoff function $f$. Then $F(-1/2)=c_f$. 
\end{restatable}
\begin{proof}
Let $N\sim F$.  
    \begin{itemize}
        \item (Case 1) If $F(-1/2)=0$, then by symmetry we have that $F(1/2)=1$. Consider the rejection region $[1/2,\infty)$ to test $T(N,N+1)$, which has type I error $P(N>1/2)=1-F(1/2)=0$ and type II error $P(N+1<1/2)=P(N<-1/2)=F(-1/2)=0$. We see that the tradeoff function includes the point $(1-0,0)$, implying that $T(N,N+1)=0$. This implies that $c_f=0=F(-1/2)$.
        \item (Case 2) Suppose that $F(-1/2)>0$, which implies that $F(1/2)<1$. Then, by Lemma \ref{lem:recurrence}, we have that $F(-1/2)=f(F(1/2))=f(1-F(-1/2))$ by symmetry. We see that $F(-1/2)=c_f$, since this is the unique solution to $c_f=f(1-c_f)$. 
    \end{itemize}
\end{proof}

\begin{restatable}{lemma}{lemgroup}[Lemma A.5: \citealp{awan2022log}]\label{lem:group}
    Let $F$ be a CND for a nontrivial symmetric tradeoff function $f$. Then $F(k\cdot)$ is a CND for $f^{\circ k}$ for any $k\in \ZZ^{>0}$.
\end{restatable}

\begin{restatable}{lemma}{lemCf}\label{lem:Cf}
Let $f$ be a nontrivial tradeoff function, and let $t\in \ZZ^{>0}$. Then 
\[c_{f^{\circ t}} = \begin{cases}
f^{\circ k}(c_f)&\text{ if } t=2k+1,\\
f^{\circ k}(1/2)&\text{ if } t=2k,
    \end{cases}\]    
        where $k\in \ZZ^{\geq 0}$ and the notation $f^{\circ k}$ represents $f\circ \cdots \circ f$, where $f$ appears $k$ times.
\end{restatable}
\begin{proof}
    Let $F_f$ be the constructed CND for $f$, by Proposition \ref{prop:CNDsynthetic}. Then by Lemma \ref{lem:group}, we know that $F_f(t\cdot)$ is a CND for $f^{\circ t}$. By Lemma \ref{lem:CNDcf}, we have that $F_f(-t/2)=c_{f^{\circ t}}$. Using the recursive definition in Proposition \ref{prop:CNDsynthetic}, we can also write 
    \begin{align*}
        c_{f^{\circ t}}&=F_f(-t/2)\\
        &=\begin{cases}
            F_f(-k-1/2)&\text{if } t=2k+1,\\
            F_f(-k)&\text{if } t=2k,
        \end{cases}\\
        &=\begin{cases}
            f^{\circ k}(F_f(-1/2))&\text{if } t=2k+1,\\
            f^{\circ k}(F_f(0))&\text{if } t=2k,
        \end{cases}\\
        &=\begin{cases}
            f^{\circ k}(c_f)&\text{if } t=2k+1,\\
            f^{\circ k}(1/2)&\text{if }t=2k,
        \end{cases}
    \end{align*}
    where we substituted $F_f(-1/2)=c_f$ and $F_f(0)=1/2$ by Proposition \ref{prop:CNDsynthetic}.
\end{proof}

\lemanti*
\begin{proof}
Let $a\in \RR$ be given. Let $N'\overset d =N-a$. For the hypothesis $H_0: N'$ versus $H_1: N'+t$, consider the following rejection region $(t/2,\infty)$. The type I and type II errors are $P(N'>t/2)$ and $P(N'+t\leq t/2)=P(N'\leq -t/2)$ respectively. Then 
    \begin{align}
        P(-t/2<N-a\leq t/2)&=P(-t/2<N'\leq t/2)\\
        &=1-P(N'>t/2)-P(N'\leq -t/2)\\
        &=1-\text{type I}-\text{type II}\\
        &\leq \mathrm{TV}(N',N'+t)\label{eq:TV1}\\
        &=\mathrm{TV}(N,N+t)\label{eq:TV2},
        \end{align}
         where \eqref{eq:TV1} uses the inequality: $\text{type I} + \text{type II}\geq 1-\mathrm{TV}(N',N'+t)$, and \eqref{eq:TV2} uses the fact that total variation is translation-invariant. 
\end{proof}

\begin{lemma}[Lemma A.5: \citealp{dong2022gaussian}]\label{lem:A5}
    Let $f$ and $g$ be tradeoff functions. If $T(P,Q)\geq f$ and $T(Q,R)\geq g$, then $T(P,R)\geq g\circ f$. 
\end{lemma}

\begin{lemma}\label{lem:group2}
    Let $N$ be a real-valued random variable, and call $f=T(N,N+1)$ (not necessarily symmetric). Then for $t\in \ZZ^{>0}$, $T(N,N+t)\geq f^{\circ t}$.
\end{lemma}
\begin{proof}
    We consider the sequence of tradeoff functions $T(N,N+1), T(N+1,N+2), \ldots T(N+t-1,N+t)$, which all have the tradeoff function $T(N,N+1)\geq f$. By \citet[Lemma A.5]{dong2022gaussian}, $T(N,N+t)\geq f^{\circ t}$.
\end{proof}

\thmanti*
\begin{proof}
    We begin with Lemma \ref{lem:anti}:
    \begin{align}
        \sup_{a\in \RR}P(-t/2<N-a\leq t/2) &\leq \mathrm{TV}(N,N+t)\\
        &\leq 1-2c_{f^\circ t}\label{eq:A5}\\
        &=\begin{cases}
            1-2f^{\circ k}(c_f)&\text{ if }t=2k+1,\\
            1-2f^{\circ k}(1/2)&\text{ if } t=2k,
        \end{cases}\label{eq:antiFinal}
    \end{align}
    where  \eqref{eq:A5} uses Lemma \ref{lem:group2} to lower bound $T(N,N+t)\geq f^{\circ t}$ and property 2 of Lemma \ref{lem:Cfzero} to get $\mathrm{TV}(N,N+t)\leq 1-2c_{f^\circ t}$. Finally, \eqref{eq:antiFinal} applies Lemma \ref{lem:Cf}.
\end{proof}

\lemCNDconcentration*
\begin{proof}
   We calculate
    \begin{align*}
        P(|N|\leq t/2)&=P(-t/2\leq N\leq t/2)\\
        &=1-2F_N(-t/2)\\
        &=1-2\begin{cases}
            f^{\circ k}(F_N(-1/2))&\text{ if }t=2k+1,\\
            f^{\circ k}(F_N(0))&\text{ if }t=2k,
        \end{cases}\\
        &=\begin{cases}
             1-2f^{\circ k}(c_f)&\text{ if }t=2k+1,\\
    1-2f^{\circ k}(1/2)&\text{ if }t=2k,\end{cases}
    \end{align*}
    where in the last line we used Lemma \ref{lem:CNDcf} to replace $F_N(-1/2)=c_f$, and symmetry to justify that $F_N(0)=1/2$.
\end{proof}

\corCNDoptimal*
\begin{proof}
    The result follows by combining Lemma \ref{lem:CNDconcentration} with Theorem \ref{thm:anti}.
\end{proof}

\corCNDratio*
\begin{proof}
    We reparametrize the problem in terms of $s=2t$ and consider the supremum of the ratio:
    \begin{align*}
        \sup_{s\in \RR^{\geq 0}} \left(\frac{\sup_{a\in \RR} P(-s/2<N'\leq s/2)}{P(|N|\leq (s+1)/2)}\right)
        &\leq \sup_{s\in \RR^{\geq 0}} \left(\frac{\sup_{a\in \RR} P(-\lceil s\rceil/2<N'\leq \lceil s \rceil/2)}{P(|N|\leq \lceil s\rceil/2)}\right)\\
        &\leq \sup_{n\in \ZZ^{\geq 0}}\left(\frac{\sup_{a\in \RR} P(-n/2<N'\leq n/2)}{P(|N|\leq n/2)}\right)\\
        &\leq 1,
    \end{align*}
    where we used the facts that $s\leq \lceil s\rceil$ and $s+1\geq \lceil s\rceil$, we reparametrized the supremum in terms of $n=\lceil s\rceil$, and the final inequality follows from Corollary \ref{cor:CNDoptimal}.
\end{proof}

Recall that a random variable $X$ is called $(\sigma^2,b)$-sub-exponential with mean $\mu$ if $\EE\exp(\lambda(X-\mu))\leq \exp(\sigma^2\lambda^2/2)$ for all $|\lambda|<1/b$. Lemma \ref{lem:Bernstein} is a sufficient condition to to establish that a random variable is sub-exponential. 

\begin{lemma}
    [Bernstein Condition]\label{lem:Bernstein}
    If $\EE |X-\mu|^k\leq \frac 12 \sigma^2 b^{k-2} k!$, for all $k\geq 2$, then 
    \begin{enumerate}
    \item $\EE\exp(\lambda (X-\mu))\leq \exp\left(\frac{\lambda^2 \sigma^2}{2(1-b|\lambda|)}\right)$ for all $|\lambda|<1/b$, and
        \item $X$ is $(2\sigma^2,2b)$-sub-exponential.
    \end{enumerate}
\end{lemma}
\begin{proof}
    Part 1 is proved in \citet[Corollary 5.2]{zhang2021concentration}. For part 2, let $|\lambda|<1/(2b)$. Then, $1/(1-b|\lambda|)<2$. With part 1, we have 
    $\EE\exp(\lambda(X-\mu))<\exp\left(\sigma^2\lambda^2\right)$. 
\end{proof}

\thmsubexponential*
\begin{proof}
    \begin{enumerate}
    \item Recall, by part 3 of Lemma \ref{lem:Cfzero}, that $f\leq f_{\ep_f,0}$, where $\ep_f=\log\left(\frac{1-c_f}{c_f}\right)$. Then, $f(1/2)\leq f_{\ep_f,0}(1/2)=(1/2)\exp(-\ep_f)$, since $c_f\leq 1/2$. Recall that by Lemma \ref{lem:recurrence} and symmetry, all CNDs satisfy $F_N(-z)=f^{\circ z}(F(0))=f^{\circ z}(1/2)$ for $z\in \ZZ^{>0}$.  So, 
    \[ F_N(-z)=f^{\circ z}(1/2)
        \leq f_{\ep_f,0}^{\circ z}(1/2)
        =(1/2)(\exp(-\ep_f))^z
        =(1/2)\exp(-\ep_f z).\]  
    
    Now let $t\geq 0$ (not necessarily an integer). Then,
    \begin{align}
        P(|N|>t)&\leq P(|N|>\lfloor t\rfloor)\\
        &=2F_N(-\lfloor t\rfloor)\label{eq:sym}\\
        &\leq 2(1/2)\exp(-\ep_f \lfloor t\rfloor)\label{eq:earlier}\\
        &\leq \exp(-\ep_f (t-1)),
    \end{align}
    where \eqref{eq:sym} used symmetry of CNDs and \eqref{eq:earlier} used our earlier result.
    \item Next, we calculate for $n\in \ZZ^{>0}$,
    \begin{align}
        \EE|N|^n&=\int_0^\infty n x^{n-1} P(|N|>x) \ dx\label{eq:Ross}\\
        &\leq \int_0^\infty nx^{n-1}  \exp(-\ep_f (x-1))\ dx\label{eq:part1}\\
        &=\ep_f^{-1}\exp(\ep_f)n \EE X^{n-1},\quad \text{where } X\sim \mathrm{Exp}(\ep_f)\\
        &=\ep_f^{-1}\exp(\ep_f)n\cdot(n-1)!\ep_f^{-(n-1)}\\
        &=\ep_f^{-n}\exp(\ep_f)n!,
    \end{align}
    where \eqref{eq:Ross} uses \citet[Self-Test Exercise 7.20]{ross2019first} and \eqref{eq:part1} uses the result from part 1. 
    \item Recall that the Bernstein condition for mean-zero sub-exponential random variables states that if 
    $\EE |N|^n\leq (1/2) n! \sigma^2 b^{n-2},$ 
    then $N$ is $( 2\sigma^2, 2 b)$-sub-exponential. Comparing our result with this, we identify $\sigma^2=2\exp(\ep_f)/\ep_f^2$ and $b=1/\ep_f$ and we have that the CND $N$ is $(4\exp(\ep_f)/\ep_f^2, 2/\ep_f)$-sub-exponential. 
    \end{enumerate}
\end{proof}

\lemlogConcaveConstruction*
\begin{proof}
    Because $F$ is a log-concave CND for $f_1$, we have that $F_t(\cdot)\defeq F(t\cdot)$ is a CND for $f_t$ \citep{awan2022log}. Since $F(0)=1/2$, we have that $F(-t)=F_t(-1)=f_t(1/2)$. It remains to establish the connection with $c_{f_{2t}}$. 

    Case 1: If $f_t(1/2)=0$, then by symmetry of $f_t$, we have that $f_t(1)\leq 1/2$. This is because if $f=T(P,Q)=T(Q,P)$, then the type I  error $\alpha$ and type II error $\beta$ can be interchanged: $\beta=f(1-\alpha)$ if and only if $\alpha=f(1-\beta)$. It follows that 
    \[f_{2t}(1-0)=f_t^{\circ 2}(1)\leq f_t(1/2)=0.\]
    Since tradeoff functions are non-negative, we have that $f_{2t}(1-0)=0$ and conclude $c_{f_{2t}}=0=f_{t}(1/2)$. 

    Case 2: If $f_t(1/2)>0$, then consider the following:
    \begin{align}
        f_{2t}(1-f_t(1/2))
        &=f_t^{\circ 2}(1-F_t(-1))\\
        &=f_t^{\circ 2}(F_t(1))\\
        &=F_t(-1)\label{eq:logCf}\\
        &=F(-t)\\
        &=f_t(1/2),
    \end{align}
    where in \eqref{eq:logCf}, we use the fact that $F_t(1)=1-F_t(-1)=1-F(-t)=1-f_{t}(1/2)<1$ and apply the second recursion formula from Lemma \ref{lem:recurrence}. 
\end{proof}

\thmstochLog*
\begin{proof}
Let $t\in \RR^{\geq 0}$ and let $a\in \RR$. Then, 
    \begin{align}
        P(|N'-a|<t)&\leq P(-t<N'-a\leq t)\label{eq:log1}\\
        &\leq \mathrm{TV}(N',N'+2t)\label{eq:log2}\\
        &\leq 1-2c_{f_{2t}}\label{eq:log3}\\
        &=1-2F_N(-t)\label{eq:log4}\\
        &=-F(-t)+F(t)\label{eq:log5}\\
        &=P(|N|<t),\label{eq:log6}
    \end{align}
    where in \eqref{eq:log2} we apply Lemma \ref{lem:anti}, in \eqref{eq:log3} we use the assumption that $T(N',N'+2t)\geq f_{2t}$ which implies that $\mathrm{TV}(N',N'+2t)\leq 1-2c_{f_{2t}}$, in \eqref{eq:log4} we use Lemma \ref{lem:logConcaveConstruction}, and \eqref{eq:log5} and \eqref{eq:log6} use the symmetry and continuity of $F_N$. 
    Finally, we have that 
    \[P(|N'-a|\leq t)=\lim_{s\downarrow t} P(|N'-a|<s)
    \leq \lim_{s\downarrow t} P(|N|<s)=P(|N|\leq t).\]
    The other statements follow as standard properties of stochastic dominance. 
\end{proof}

\propround*
\begin{proof}
Property 1 of Definition \ref{def:dCND} follows from fact that $T(\Delta N_c,\Delta N_c+\Delta)\geq f$ and the postprocessing property of tradeoff functions. For property 2 of Definition \ref{def:dCND}, note that 
\begin{align*}
    P(N=x)&=P(x-1/2\leq \Delta N_c< x+1/2)\\
    &=F_{\Delta N_c}(x+1/2)-F_{\Delta N_c}(x-1/2),
\end{align*}
since $N_c$ is a  continuous random variable. It follows that $F_N(t) =P(N\leq t) = F_{\Delta N_c}(t+1/2)$ for all integers $t$. By Lemma \ref{lem:recurrence}, we know that $f(F_{N_c}(t+1))=F_{N_c}(t)$ for all $t\in \RR$ such that $F_{N_c}(t+1)<1$, or equivalently, $f(F_{\Delta N_c}(t+\Delta ))=F_{\Delta N_c}(t)$ for all $t\in \RR$ such that $F_{\Delta N_c}(t+\Delta )<1$. Property 3 follows immediately from the symmetry of $N_c$.  
\end{proof}

\propunique*
\begin{proof}
 First we will show that properties 2 and 3 of Definition \ref{def:dCND} uniquely determine the cdf, so there is at most one discrete CND. Notice that given $F(0)$, the recursion in property 2 of Definition \ref{def:dCND} fully determines the cdf. Property 3 of Definition \ref{def:dCND} implies that $F(-1)=P(N<0)=P(N>0)=1-F(0)$. Combining this with $f(F(0))=F(-1)$, which is from property 2, we have that $f(F(0))=1-F(0)$, which implies that $F(0)=1-c_f$. We see that there is at most one discrete CND at sensitivity one. Since Proposition \ref{prop:round} established that the rounding of any continuous CND gives a discrete CND, it follows that this construction results in the unique discrete CND at sensitivity 1. 

Let $N_c$ be any CND for $f$. Note that $F_{N_c}(1/2)=1-c_f$ and $F_{N_c}(-1/2)=c_f$, by Lemma \ref{lem:CNDcf}. By the recursion in Lemma \ref{lem:recurrence} it follows that for an integer $x<=0$, $F_{N_c}(x+1/2)=f^{\circ x}(1-c_f)$ and $F_{N_c}(x-1/2) = f^{\circ x}(c_f)$. Using the expression from the proof of Proposition \ref{prop:round}, we have that for an integer $x\leq 0$, 
\begin{align*}
P(N=x-1) = F_{N_c}(x-1+1/2)-F_{N_c}(x-1-1/2)=f^{\circ x}(1-c_f)-f^{\circ x}(c_f).
\end{align*}
Finally, we have that $P(N=x)$ is symmetric with the desired formula:
\begin{align*}
P(N=-x) &=F_{ N_c}(-x+1/2)-F_{ N_c}(-x-1/2)\\
&=1-F_{N_c}(x-1/2)-1+F_{N_c}(x+1/2)\\
&=F_{N_c}(x+1/2)-F_{N_c}(x-1/2)\\
&=P(N=x),
\end{align*}
where we used the fact that $F_{N_c}$ is symmetric (i.e., $F(-x)=1-F(x)$).
\end{proof}

\corinteger*
\begin{proof}
Since $N$ takes integer values, we can write
\begin{align*}
\sup_{a\in \ZZ} P(|N-a|\leq t)=\sup_{a\in \ZZ} P\Big(-(2t+1)/2<N-a\leq (2t+1)/2\Big)\leq 1-2f^{\circ t}(c_f),
\end{align*}
where the inquality uses Theorem \ref{thm:anti}.
\end{proof}

\lemsensOne*
\begin{proof}
Let $F_f$ be the CND constructed in Proposition \ref{prop:CNDsynthetic}, and recall that $F_N(t) = F_f(t+1/2)$ for $t\in \ZZ$, by Proposition \ref{prop:unique}. Then, 
\begin{align}
    P(|N|\leq t)&=P(-t\leq N\leq t)\\
    &=1-2F_f(-t-1/2)\label{eq:sym1}\\
    &=1-2f^{\circ t}(F(-1/2))\label{eq:recurrence1}\\
    &=1-2f^{\circ t}(c_f),\label{eq:cf1}
\end{align}
where in \eqref{eq:sym1}, we use the fact that $F_f$ is the cdf of a symmetric continuous random variable, in \eqref{eq:recurrence1} we use the recurrence in Proposition \ref{prop:CNDsynthetic}, and in \eqref{eq:cf1} we use the construction in Proposition \ref{prop:CNDsynthetic} that $F_f(-1/2)=c_f$. 
\end{proof}

\thmstochDisc*
\begin{proof}
    Let $t\in \ZZ^{\geq 0}$. Then, 
       $ \sup_{a\in \ZZ}P(|N'-a|\leq t)\leq 1-2f^{\circ t}(c_f)\label{eq:stochD1}
        =P(|N|\leq t)$, 
    where the inequality used Corollary \ref{cor:integer} and the equality used Lemma \ref{lem:sens1}. The first two bullets follow from properties of stochastic dominance. Using $a=0$, we get the third bullet.
\end{proof}


\subsection{Derivation of the Cauchy-DP tradeoff function}
In Example \ref{ex:cauchy}, we implement the CND via Proposition \ref{prop:CNDsynthetic} for $C_1=T(\mathrm{Cauchy}(0,1),\mathrm{Cauchy}(1,1))$. Sampling from this distribution is straightforward provided that we can evaluate $C_1$ and have $c_1=c_{C_1}$, using Lemma \ref{lem:F6} below:

\begin{lemma}[Proposition F.6: \citealp{awan2023canonical}]\label{lem:F6}
   Let $f$ be a symmetric nontrivial
   
   \noindent tradeoff function and let $F_f$ be as in
Proposition \ref{prop:CNDsynthetic}. Then the quantile function 
$F^{-1}_f : (0, 1) \rightarrow \RR$ for $F_f$ can be expressed as 
\[F^{-1}_f(u) = \begin{cases}
F^{-1}_f(1-f(1-u))& u<c_f,\\
\frac{u-1/2}{1-2c_f}&c_f\leq u\leq 1-c_f,\\
F^{-1}_f(f(u))+1&u>1-c_f.
\end{cases}\]
Furthermore, for any $u\in (0,1)$, the $F^{-1}_f(u)$ takes a finite number of recursive steps to evaluate. Thus, if $U\sim U(0,1)$ then $F^{-1}_f(U)\sim F_f$.
\end{lemma}

Recall that in Example \ref{ex:cauchy}, we calculated that $c_1=1/2-(1/\pi)\arctan(1/2)$, so it only remains to numerically evaluate $C_1$. By the Neyman-Pearson Lemma, we know that the family of optimal hypothesis tests reject when the likelihood ratio is above a given threshold. Given an observation $x\in \RR$ from either $\mathrm{Cauchy}(0,1)$ or $\mathrm{Cauchy}(1,1)$, the likelihood ratio statistic is 
\[\mathrm{LRT}(x) = \frac{\pi(1+x^2)}{\pi(1+(x-1)^2)}=1+\frac{2x-1}{x^2-2x+2},\]
For rejection threshold $t+1\geq 1$, we have the corresponding rejection region $[x_1(t),x_2(t)]$, where
\[x_1(t)=\frac{t+1-\sqrt{-t^2+t+1}}{t},\]
\[x_2(t) = \frac{t+1+\sqrt{-t^2+t+1}}{t}.\]
The type I error is $1-\alpha(t) = F_C(x_2(t))-F_C(x_1(t))$, where $F_C$ is the cdf of $\mathrm{Cauchy}(0,1)$, and the type II error is $C_1(\alpha(t)) = 1-F_C(x_2(t)-1)+F_C(x_1(t)-1)$. Since $\alpha(t)$ is a monotone function, given $\alpha\in [1-c,1]$ we can numerically solve for $t$ such that $\alpha(t)=\alpha$ and then evaluate $C_1(\alpha(t))$. 

So far, we have a method of evaluating $C_1(\alpha)$ for $\alpha\in [1-c,1]$. Since $\mathrm{LRT}(1-x)=1/\mathrm{LRT}(x)$, we have that for $\alpha\in [0,1-c]$, the rejection region is of the form $(-\infty,1-x_2(t)]\cup[1-x_1(t),\infty)$, where $t$ is chosen to satisfy $1-\alpha = F_C(1-x_2(t))+1-F_C(1-x_1(t))$. Then $C_1(\alpha) = F_C(-x_1(t))-F_C(-x_2(t))$. 

Code to evaluate this tradeoff function is provided, along with code to implement Lemma \ref{lem:F6} above, enabling us to sample from the CND of Proposition \ref{prop:CNDsynthetic}. 



\vskip 0.2in
\bibliography{bibliography}

\begin{thebibliography}{43}
\providecommand{\natexlab}[1]{#1}
\providecommand{\url}[1]{\texttt{#1}}
\expandafter\ifx\csname urlstyle\endcsname\relax
  \providecommand{\doi}[1]{doi: #1}\else
  \providecommand{\doi}{doi: \begingroup \urlstyle{rm}\Url}\fi

\bibitem[Abadi et~al.(2016)Abadi, Chu, Goodfellow, McMahan, Mironov, Talwar,
  and Zhang]{abadi2016deep}
Martin Abadi, Andy Chu, Ian Goodfellow, H~Brendan McMahan, Ilya Mironov, Kunal
  Talwar, and Li~Zhang.
\newblock Deep learning with differential privacy.
\newblock In \emph{Proceedings of the 2016 ACM SIGSAC conference on computer
  and communications security}, pages 308--318, 2016.

\bibitem[Abowd(2018)]{abowd2018us}
John~M Abowd.
\newblock The {US Census Bureau} adopts differential privacy.
\newblock In \emph{Proceedings of the 24th ACM SIGKDD International Conference
  on Knowledge Discovery \& Data Mining}, pages 2867--2867, 2018.

\bibitem[Agarwal et~al.(2018)Agarwal, Suresh, Yu, Kumar, and
  McMahan]{agarwal2018cpsgd}
Naman Agarwal, Ananda~Theertha Suresh, Felix Xinnan~X Yu, Sanjiv Kumar, and
  Brendan McMahan.
\newblock cp{SGD}: Communication-efficient and differentially-private
  distributed {SGD}.
\newblock \emph{Advances in Neural Information Processing Systems}, 31, 2018.

\bibitem[Agarwal et~al.(2021)Agarwal, Kairouz, and Liu]{agarwal2021skellam}
Naman Agarwal, Peter Kairouz, and Ziyu Liu.
\newblock The {Skellam} mechanism for differentially private federated
  learning.
\newblock \emph{Advances in Neural Information Processing Systems},
  34:\penalty0 5052--5064, 2021.

\bibitem[Awan and Dong(2022)]{awan2022log}
Jordan Awan and Jinshuo Dong.
\newblock Log-concave and multivariate canonical noise distributions for
  differential privacy.
\newblock \emph{Advances in Neural Information Processing Systems},
  35:\penalty0 34229--34240, 2022.

\bibitem[Awan and Slavkovi{\'c}(2018)]{awan2018differentially}
Jordan Awan and Aleksandra Slavkovi{\'c}.
\newblock Differentially private uniformly most powerful tests for binomial
  data.
\newblock \emph{Advances in Neural Information Processing Systems}, 31, 2018.

\bibitem[Awan and Vadhan(2023)]{awan2023canonical}
Jordan Awan and Salil Vadhan.
\newblock Canonical noise distributions and private hypothesis tests.
\newblock \emph{Annals of Statistics}, 51\penalty0 (2):\penalty0 547--572,
  2023.

\bibitem[Awan and Wang(2022)]{awan2022differentially}
Jordan Awan and Yue Wang.
\newblock Differentially private {Kolmogorov-Smirnov}-type tests.
\newblock \emph{arXiv preprint arXiv:2208.06236}, 2022.

\bibitem[Bun and Steinke(2016)]{bun2016concentrated}
Mark Bun and Thomas Steinke.
\newblock Concentrated differential privacy: Simplifications, extensions, and
  lower bounds.
\newblock In \emph{Theory of Cryptography Conference}, pages 635--658.
  Springer, 2016.

\bibitem[Canonne et~al.(2020)Canonne, Kamath, and Steinke]{canonne2020discrete}
Cl{\'e}ment~L Canonne, Gautam Kamath, and Thomas Steinke.
\newblock The discrete {Gaussian} for differential privacy.
\newblock \emph{Advances in Neural Information Processing Systems},
  33:\penalty0 15676--15688, 2020.

\bibitem[Chaudhuri et~al.(2011)Chaudhuri, Monteleoni, and
  Sarwate]{chaudhuri2011differentially}
Kamalika Chaudhuri, Claire Monteleoni, and Anand~D Sarwate.
\newblock Differentially private empirical risk minimization.
\newblock \emph{Journal of Machine Learning Research}, 12\penalty0 (3), 2011.

\bibitem[Dong(2020)]{dong2020gaussian}
Jinshuo Dong.
\newblock \emph{Gaussian differential privacy and related techniques}.
\newblock PhD thesis, University of Pennsylvania, 2020.

\bibitem[Dong et~al.(2022)Dong, Roth, and Su]{dong2022gaussian}
Jinshuo Dong, Aaron Roth, and Weijie~J Su.
\newblock Gaussian differential privacy.
\newblock \emph{Journal of the Royal Statistical Society Series B: Statistical
  Methodology}, 84\penalty0 (1):\penalty0 3--37, 2022.

\bibitem[Du et~al.(2020)Du, Foot, Moniot, Bray, and
  Groce]{du2020differentially}
Wenxin Du, Canyon Foot, Monica Moniot, Andrew Bray, and Adam Groce.
\newblock Differentially private confidence intervals.
\newblock \emph{arXiv preprint arXiv:2001.02285}, 2020.

\bibitem[Dwork et~al.(2006{\natexlab{a}})Dwork, Kenthapadi, McSherry, Mironov,
  and Naor]{dwork2006our}
Cynthia Dwork, Krishnaram Kenthapadi, Frank McSherry, Ilya Mironov, and Moni
  Naor.
\newblock Our data, ourselves: Privacy via distributed noise generation.
\newblock In \emph{Advances in Cryptology-EUROCRYPT 2006: 24th Annual
  International Conference on the Theory and Applications of Cryptographic
  Techniques, St. Petersburg, Russia, May 28-June 1, 2006. Proceedings 25},
  pages 486--503. Springer, 2006{\natexlab{a}}.

\bibitem[Dwork et~al.(2006{\natexlab{b}})Dwork, McSherry, Nissim, and
  Smith]{dwork2006calibrating}
Cynthia Dwork, Frank McSherry, Kobbi Nissim, and Adam Smith.
\newblock Calibrating noise to sensitivity in private data analysis.
\newblock In \emph{Theory of cryptography conference}, pages 265--284.
  Springer, 2006{\natexlab{b}}.

\bibitem[Erlingsson et~al.(2014)Erlingsson, Pihur, and
  Korolova]{erlingsson2014rappor}
{\'U}lfar Erlingsson, Vasyl Pihur, and Aleksandra Korolova.
\newblock Rappor: Randomized aggregatable privacy-preserving ordinal response.
\newblock In \emph{Proceedings of the 2014 ACM SIGSAC conference on computer
  and communications security}, pages 1054--1067, 2014.

\bibitem[Feldman and Zrnic(2021)]{feldman2021individual}
Vitaly Feldman and Tijana Zrnic.
\newblock Individual privacy accounting via a {R\'enyi} filter.
\newblock \emph{Advances in Neural Information Processing Systems},
  34:\penalty0 28080--28091, 2021.

\bibitem[Geng and Viswanath(2015{\natexlab{a}})]{geng2015approx}
Quan Geng and Pramod Viswanath.
\newblock Optimal noise adding mechanisms for approximate differential privacy.
\newblock \emph{IEEE Transactions on Information Theory}, 62\penalty0
  (2):\penalty0 952--969, 2015{\natexlab{a}}.

\bibitem[Geng and Viswanath(2015{\natexlab{b}})]{geng2015optimal}
Quan Geng and Pramod Viswanath.
\newblock The optimal noise-adding mechanism in differential privacy.
\newblock \emph{IEEE Transactions on Information Theory}, 62\penalty0
  (2):\penalty0 925--951, 2015{\natexlab{b}}.

\bibitem[Ghosh et~al.(2012)Ghosh, Roughgarden, and
  Sundararajan]{ghosh2012universally}
Arpita Ghosh, Tim Roughgarden, and Mukund Sundararajan.
\newblock Universally utility-maximizing privacy mechanisms.
\newblock \emph{SIAM Journal on Computing}, 41\penalty0 (6):\penalty0
  1673--1693, 2012.

\bibitem[Gupte and Sundararajan(2010)]{gupte2010universally}
Mangesh Gupte and Mukund Sundararajan.
\newblock Universally optimal privacy mechanisms for minimax agents.
\newblock In \emph{Proceedings of the twenty-ninth ACM SIGMOD-SIGACT-SIGART
  symposium on Principles of database systems}, pages 135--146, 2010.

\bibitem[Kairouz et~al.(2016)Kairouz, Oh, and Viswanath]{kairouz2016extremal}
Peter Kairouz, Sewoong Oh, and Pramod Viswanath.
\newblock Extremal mechanisms for local differential privacy.
\newblock \emph{The Journal of Machine Learning Research}, 17\penalty0
  (1):\penalty0 492--542, 2016.

\bibitem[Kazan et~al.(2023)Kazan, Shi, Groce, and Bray]{kazan2023test}
Zeki Kazan, Kaiyan Shi, Adam Groce, and Andrew~P Bray.
\newblock The test of tests: A framework for differentially private hypothesis
  testing.
\newblock In \emph{International Conference on Machine Learning}, pages
  16131--16151. PMLR, 2023.

\bibitem[Kemp(1997)]{kemp1997characterizations}
Adrienne~W Kemp.
\newblock Characterizations of a discrete normal distribution.
\newblock \emph{Journal of Statistical Planning and Inference}, 63\penalty0
  (2):\penalty0 223--229, 1997.

\bibitem[Kifer et~al.(2012)Kifer, Smith, and Thakurta]{kifer2012private}
Daniel Kifer, Adam Smith, and Abhradeep Thakurta.
\newblock Private convex empirical risk minimization and high-dimensional
  regression.
\newblock In \emph{Conference on Learning Theory}, pages 25--1. JMLR Workshop
  and Conference Proceedings, 2012.

\bibitem[Koskela et~al.(2023)Koskela, Tobaben, and
  Honkela]{koskela2023individual}
Antti Koskela, Marlon Tobaben, and Antti Honkela.
\newblock Individual privacy accounting with {Gaussian} differential privacy.
\newblock In \emph{The Eleventh International Conference on Learning
  Representations}, 2023.

\bibitem[Krishnapur(2016)]{krishnapur2016anti}
Manjunath Krishnapur.
\newblock Anti-concentration inequalities.
\newblock \emph{Lecture notes}, 2016.

\bibitem[Levy(1992)]{levy1992stochastic}
Haim Levy.
\newblock Stochastic dominance and expected utility: Survey and analysis.
\newblock \emph{Management science}, 38\penalty0 (4):\penalty0 555--593, 1992.

\bibitem[McSherry and Talwar(2007)]{mcsherry2007mechanism}
Frank McSherry and Kunal Talwar.
\newblock Mechanism design via differential privacy.
\newblock In \emph{48th Annual IEEE Symposium on Foundations of Computer
  Science (FOCS'07)}, pages 94--103. IEEE, 2007.

\bibitem[Mironov(2012)]{mironov2012significance}
Ilya Mironov.
\newblock On significance of the least significant bits for differential
  privacy.
\newblock In \emph{Proceedings of the 2012 ACM conference on Computer and
  communications security}, pages 650--661, 2012.

\bibitem[Mironov(2017)]{mironov2017renyi}
Ilya Mironov.
\newblock {R{\'e}nyi} differential privacy.
\newblock In \emph{2017 IEEE 30th computer security foundations symposium
  (CSF)}, pages 263--275. IEEE, 2017.

\bibitem[Nielsen and Okamura(2022)]{nielsen2022f}
Frank Nielsen and Kazuki Okamura.
\newblock On f-divergences between {Cauchy} distributions.
\newblock \emph{IEEE Transactions on Information Theory}, 2022.

\bibitem[Qin et~al.(2022{\natexlab{a}})Qin, He, Fang, and
  Lam]{qin2022differential}
Shuying Qin, Jianping He, Chongrong Fang, and James Lam.
\newblock Differential private discrete noise adding mechanism: Conditions,
  properties and optimization.
\newblock \emph{arXiv preprint arXiv:2203.10323}, 2022{\natexlab{a}}.

\bibitem[Qin et~al.(2022{\natexlab{b}})Qin, He, Fang, and
  Lam]{qin2022differential2}
Shuying Qin, Jianping He, Chongrong Fang, and James Lam.
\newblock Differential private discrete noise adding mechanism: Conditions and
  properties.
\newblock In \emph{2022 American Control Conference (ACC)}, pages 946--951.
  IEEE, 2022{\natexlab{b}}.

\bibitem[Reimherr and Awan(2019)]{reimherr2019elliptical}
Matthew Reimherr and Jordan Awan.
\newblock Elliptical perturbations for differential privacy.
\newblock \emph{Advances in Neural Information Processing Systems}, 32, 2019.

\bibitem[Rogers et~al.(2016)Rogers, Roth, Ullman, and
  Vadhan]{rogers2016privacy}
Ryan~M Rogers, Aaron Roth, Jonathan Ullman, and Salil Vadhan.
\newblock Privacy odometers and filters: Pay-as-you-go composition.
\newblock \emph{Advances in Neural Information Processing Systems}, 29, 2016.

\bibitem[Ross(2019)]{ross2019first}
Sheldon Ross.
\newblock \emph{First Course in Probability, A}.
\newblock Pearson Higher Ed, 10th edition, 2019.

\bibitem[Roy(2003)]{roy2003discrete}
Dilip Roy.
\newblock The discrete normal distribution.
\newblock \emph{Communications in Statistics-Theory and Methods}, 32\penalty0
  (10):\penalty0 1871--1883, 2003.

\bibitem[Soria-Comas and Domingo-Ferrer(2013)]{SORIACOMAS2013200}
Jordi Soria-Comas and Josep Domingo-Ferrer.
\newblock Optimal data-independent noise for differential privacy.
\newblock \emph{Information Sciences}, 250:\penalty0 200--214, 2013.
\newblock ISSN 0020-0255.

\bibitem[Szab{\l}owski(2001)]{szablowski2001discrete}
Pawe{\l}~J Szab{\l}owski.
\newblock Discrete normal distribution and its relationship with {Jacobi} theta
  functions.
\newblock \emph{Statistics \& probability letters}, 52\penalty0 (3):\penalty0
  289--299, 2001.

\bibitem[Tang et~al.(2017)Tang, Korolova, Bai, Wang, and Wang]{tang2017privacy}
Jun Tang, Aleksandra Korolova, Xiaolong Bai, Xueqiang Wang, and Xiaofeng Wang.
\newblock Privacy loss in {Apple}'s implementation of differential privacy on
  {macOS} 10.12.
\newblock \emph{arXiv preprint arXiv:1709.02753}, 2017.

\bibitem[Zhang and Chen(2021)]{zhang2021concentration}
Huiming Zhang and Songxi Chen.
\newblock Concentration inequalities for statistical inference.
\newblock \emph{Communications in Mathematical Research}, 37\penalty0
  (1):\penalty0 1--85, 2021.
\newblock ISSN 2707-8523.

\end{thebibliography}

\end{document}